\documentclass[11pt,a4paper,reqno,intlimits]{amsart}

\usepackage{amssymb}
\usepackage{chicago}
\usepackage{color}
\usepackage[mathscr]{euscript}
\usepackage{xspace}
\usepackage{enumerate}
\usepackage{epsfig}

\DeclareMathAlphabet{\mathpzc}{OT1}{pzc}{m}{it}

\begin{document}

%% ============================== DEFINITIONS =============================== %%
%% ========================================================================== %%
\theoremstyle{plain}
\newtheorem{theorem}{Theorem}[section]
\newtheorem{lemma}[theorem]{Lemma}
\newtheorem{proposition}[theorem]{Proposition}
\newtheorem{claim}[theorem]{Claim}
\newtheorem{corollary}[theorem]{Corollary}
\newtheorem{axiom}{Axiom}

\theoremstyle{definition}
\newtheorem{remark}[theorem]{Remark}
\newtheorem{note}{Note}[section]
\newtheorem{definition}[theorem]{Definition}
\newtheorem{example}[theorem]{Example}
\newtheorem*{ackn}{Acknowledgements}
\newtheorem{assumption}{Assumption}
\newtheorem{approach}{Approach}
\newtheorem{critique}{Critique}
\newtheorem{question}{Question}
\newtheorem{aim}{Aim}
\newtheorem*{asa}{Assumption ($\mathbf{A}$)}
\newtheorem*{asp}{Assumption ($\mathbb{P}$)}
\newtheorem*{ass}{Assumption ($\mathbb{S}$)}
%% ========================================================================== %%
%% ========================================================================== %%
\renewcommand{\theequation}{\thesection.\arabic{equation}}
\numberwithin{equation}{section}

\newcommand{\Law}{\ensuremath{\mathop{\mathrm{Law}}}}
\newcommand{\loc}{{\mathrm{loc}}}

\let\SETMINUS\setminus
\renewcommand{\setminus}{\backslash}

\def\stackrelboth#1#2#3{\mathrel{\mathop{#2}\limits^{#1}_{#3}}}

\newcommand{\prozess}[1][L]{{\ensuremath{#1=(#1_t)_{0\le t\le T}}}\xspace}
\newcommand{\prazess}[1][L]{{\ensuremath{#1=(#1_t)_{t\ge0}}}\xspace}
\newcommand{\pt}[1][N]{\ensuremath{\P_{T_{#1}}}\xspace}
\newcommand{\tk}[1][N]{\ensuremath{T_{#1}}\xspace}
\newcommand{\dd}[1][]{\ensuremath{\ud{#1}}\xspace}

\newcommand{\scal}[2]{\ensuremath{\langle #1, #2 \rangle}}
\newcommand{\set}[1]{\ensuremath{\left\{#1\right\}}}

%\newcommand{\R}[1][\R]{\ensuremath{\R^{#1}}\xspace}
% \newcommand{\pt}[1][]{{\ensuremath{\P_{T^*_{#1}}}}\xspace}
% \newcommand{\ts}[1][]{\ensuremath{T^*_{#1}}\xspace}
%% ========================================================================== %%
%% ========================================================================== %%
\def\lev{L\'{e}vy\xspace}
\def\lk{L\'{e}vy--Khintchine\xspace}
\def\lib{LIBOR\xspace}
\def\mg{martingale\xspace}
\def\smmg{semimartingale\xspace}

\def\half{\frac{1}{2}}

\def\F{\ensuremath{\mathcal{F}}}
\def\bF{\mathbf{F}}
\def\R{\ensuremath{\mathbb{R}}}
\def\Rp{\mathbb{R}_{\geqslant0}}
\def\Rm{\mathbb{R}_{\leqslant 0}}
\def\C{\ensuremath{\mathbb{C}}}
\def\U{\ensuremath{\mathcal{U}}}
\def\I{\mathcal{I}}

\def\ii{\ensuremath{\mathrm{i}}}

\def\P{\ensuremath{\mathrm{I\kern-.2em P}}}
\def\E{\ensuremath{\mathrm{I\kern-.2em E}}}

\def\ott{{0\leq t\leq T}}
\def\idd{{1\le i\le d}}

\def\icc{\mathpzc{i}}
\def\ecc{\mathbf{e}_\mathpzc{i}}

\def\uk{u_{k+1}}

\def\e{\mathrm{e}}
\def\ud{\ensuremath{\mathrm{d}}}
\def\dt{\ud t}
\def\ds{\ud s}
\def\dx{\ud x}
\def\dy{\ud y}
\def\dv{\ud v}
\def\dsdx{\ensuremath{(\ud s, \ud x)}}
\def\dtdx{\ensuremath{(\ud t, \ud x)}}

\def\lsnc{\ensuremath{\mathrm{LSNC-}\chi^2}}
\def\nc{\ensuremath{\mathrm{NC-}\chi^2}}

\newcommand{\Ind}[1]{\mathbf{1}\left\{#1\right\}}
%% ========================================================================== %%
%% ========================================================================== %%

\title{The affine LIBOR models}

\author[M. Keller-Ressel, A. Papapantoleon, J. Teichmann]
       {Martin Keller-Ressel, Antonis Papapantoleon,\\Josef Teichmann}

\address{Institute of Mathematics, TU Berlin, Stra\ss e des 17. Juni 136,
         10623 Berlin, Germany}
\email{mkeller@math.tu-berlin.de}

\address{Institute of Mathematics, TU Berlin, Stra\ss e des 17. Juni 136,
         10623 Berlin, Germany}
\email{papapan@math.tu-berlin.de}

\address{ETH Z\"urich, Departement Mathematik, R\"amistrasse 101,
         CH-8092 Z\"urich, Switzerland}
\email{josef.teichmann@math.ethz.ch}

\thanks{All authors gratefully acknowledge the financial support from the
        FWF-grant Y 328 (START prize of the Austrian Science Fund).
        We would also like to thank Ismail Laachir (INRIA Rocquencourt) for
        helpful remarks. A previous version was circulated under the title
        ``\textit{A new approach to LIBOR modeling}''.}


\date{}\maketitle\pagestyle{myheadings}\frenchspacing

\begin{abstract}
We provide a general and flexible approach to LIBOR modeling based on the class
of affine factor processes. Our approach respects the basic economic requirement
that LIBOR rates are non-negative, and the basic requirement from mathematical
finance that LIBOR rates are analytically tractable martingales with respect to
their own forward measure. Additionally, and most importantly, our approach also
leads to analytically tractable expressions of multi-LIBOR payoffs. This
approach unifies therefore the advantages of well-known forward price models
with those of classical LIBOR rate models. Several examples are added and
prototypical volatility smiles are shown. We believe that the CIR-process based
LIBOR model might be of particular interest for applications, since closed form
valuation formulas for caps and swaptions are derived.
\end{abstract}

\section{Introduction}

Let $T_0<\ldots<T_N$ be a discrete tenor of maturity dates. LIBOR rates are
related to the observable ratio of prices of zero-coupon bonds with maturity
$T_{k-1}$ and $T_{k}=T_{k-1}+\delta$ via
\begin{align*}
L(t,T_{k-1}) = \frac{1}{\delta} \left(\frac{B(t,T_{k-1})}{B(t,T_k)}-1\right).
\end{align*}
The nature of interbank loans, as well as the daily calculation of LIBOR rates
as the trimmed arithmetic average of interbank quoted rates  (see
\texttt{www.bbalibor.com}), yield that LIBOR rates should be non-negative. A
requirement from mathematical finance is that LIBOR rates should be martingales
with respect to their own forward measure $ \P_{T_k} $. That is, when
$B(\cdot,T_{k}) $ is considered as num\'eraire of the model, then discounted
bond prices $\big(\frac{B(t,T_{k-1})}{B(t,T_k)}\big), {0\leq t\leq T_{k-1}}$
should be martingales. An additional basic requirement motivated by applications
is the tractability of the model, since otherwise one cannot calibrate to the
market data. Therefore the LIBOR rate processes $ L(\cdot,T_{k-1}) $ should have
tractable stochastic dynamics with respect to their forward measure $\P_{T_k}$,
for $k=1,\ldots,N$; for instance of exponential L\'evy type along the discrete
tenor of dates $T_0<\ldots<T_N $. Here the terminus ``analytically tractable''
is used in the sense that either the density of the stochastic factors driving
the LIBOR rate process is known explicitly, or the characteristic function. In
both cases, the numerical evaluation, which is needed for calibration to the
market, is easily done.

In applications, the stochastic factors have to be evaluated with respect to
different num\'eraires. In order to describe the dynamics with respect to a
suitable martingale measure, for instance the terminal forward measure
$\P_{T_N}$, we have to perform a change of measure. Usually this change of
measure destroys the tractable structure of $ L(\cdot,T_{k-1}) $ with
respect to its forward measure. This well-known phenomenon makes LIBOR market
models based on Brownian motions or L\'evy processes quite delicate to apply for
multi-LIBOR-dependent payoffs: either one performs expensive Monte Carlo
simulations or one has to approximate the equation (the keyword here is
``freezing the drift'', see e.g. \citeNP{SiopachaTeichmann07}).

In order to overcome this natural intractability, forward price models have been
considered, where the tractability with respect to other forward measures is
pertained when changing the measure. Hence, modeling forward prices
$F(\cdot,T_{k-1},T_{k})=1+\delta L(\cdot,T_k)$ produces a very tractable model
class. However negative LIBOR rates can occur with positive probability, which
contradicts any economic intuition.

In this work, we propose a new approach to modeling LIBOR rates based on affine
processes. The approach follows the footsteps of forward price models, however,
we are able to circumvent their drawback: in our approach LIBOR rates are almost
surely non-negative. Moreover, the model remains analytically tractable with
respect to all possible forward measures, hence the calibration and evaluation
of derivatives is fairly simple. In fact, this is the first LIBOR model
where the following are satisfied simultaneously:
\begin{itemize}
\item \lib rates are non-negative;
\item caps and swaptions can be priced easily using Fourier methods, for several
      affine factor processes;
\item closed-form valuation formulas for caps \textit{and} swaptions are derived
      for the CIR process, in 1- and 2-factor models.
\end{itemize}

A particular feature of our approach is that the factor process is a
time-homogenous Markov process when we consider the model with respect to the
terminal measure $\P_{T_N}$. With respect to forward measures the factor
processes will show time-inhomogeneities due to the nature of the change of
measure. When we compare our approach to an affine factor setting within the
HJM-methodology, we observe that in both cases one can choose -- with respect to
the spot measure in the HJM setting or with respect to the terminal measure in
our setting -- a time-homogeneous factor process. \lib rates have in both cases
a typical dependence on time-to-maturity $ T_N -t $.

The remainder of the article is organized as follows: in Section \ref{axioms} we
formulate basic axioms for LIBOR market models. In Section \ref{ex_approaches}
we recapitulate the literature on LIBOR models. In Section \ref{affine} we
introduce affine processes which are applied in Section \ref{MMge1} to the
construction of certain martingales. In Section \ref{new_approach} we present
our new approach to LIBOR market models, which is applied in Section
\ref{derivatives} to derivative pricing. In Section \ref{toy-1} several
examples, including the CIR-based models, are presented and in Section
\ref{numerical_illustrations} we show prototypical volatility surfaces generated
by the models.

\section{Axioms}
\label{axioms}

Let us denote by $L(t,T)$ the time-$t$ forward \lib rate that is settled at time
$T$ and received at time $T+\delta$; here $T$ denotes some finite time horizon.
The \lib rate is related to the prices of zero coupon bonds, denoted by
$B(t,T)$, and the forward price, denoted by $F(t,T,T+\delta)$, by the following
equations:
\begin{align}\label{basic}
  1+\delta L(t,T) = \frac{B(t,T)}{B(t,T+\delta)} = F(t,T,T+\delta).
\end{align}
One postulates that the \lib rate should satisfy the following axioms, motivated
by \textit{economic theory}, \textit{arbitrage pricing theory} and
\textit{applications}.

\begin{axiom}
The \lib rate should be \emph{non-negative}, i.e. $L(t,T)\ge0$ for all
$0\le t\le T$.
\end{axiom}
\begin{axiom}
The \lib rate process should be a \emph{martingale} under the corresponding
\emph{forward measure}, i.e. $L(\cdot,T)\in\mathcal{M}(\P_{T+\delta})$.
\end{axiom}
\begin{axiom}
The \lib rate process, i.e. the (multivariate) collection of all \lib rates,
should be \emph{analytically tractable} with respect to as many forward measures
as possible. Minimally, closed-form or semi-analytic valuation formulas should
be available for the most liquid interest rate derivatives, i.e. caps and
swaptions, so that the model can be \textit{calibrated} to market data in
\textit{reasonable time}.
\end{axiom}

Furthermore we wish to have \textit{rich structural properties}: that is, the
model should be able to reproduce the observed phenomena in interest rate
markets, e.g. the shape of the implied volatility surface in cap markets or the
implied correlation structure in swaption markets.

\section{Existing approaches}
\label{ex_approaches}

There are several approaches to  \lib modeling developed in the literature
attempting to fulfill the axioms and practical requirements discussed in the
previous section.  We briefly describe below the two main approaches, namely
the \textit{LIBOR market model} (LMM) and the \textit{forward price model}, and
comment on their ability to fulfill them. We also briefly discuss
\emph{Markov-functional models}.

\begin{approach}
In \emph{\lib market models}, developed in a series of articles by
\shortciteN{SandmannSondermannMiltersen95},
\shortciteN{MiltersenSandmannSondermann97}, \shortciteN{BraceGatarekMusiela97},
and \citeN{Jamshidian97}, each forward \lib rate is modeled as an exponential
Brownian motion under its corresponding forward measure. This model provides a
theoretical justification for the common market practice of pricing caplets
according to Black's futures formula \cite{Black76}, i.e. assuming that the
forward LIBOR rate is log-normally distributed. Several extensions of this
framework have been proposed in the literature, using jump-diffusions, \lev
processes or general semimartingales as the driving motion (cf. e.g.
\citeNP{GlassermanKou03}, \citeNP{EberleinOezkan05}, \citeNP{Jamshidian99}), or
incorporating stochastic volatility effects (cf. e.g. Andersen and
Brotherton-Ratcliffe \citeyearNP{AndersenBrothertonRatcliffe05}).

We can generically describe LIBOR market models as follows: on a stochastic
basis consider a discrete tenor of dates $(T_k)_{0\le k\le N}$, forward measures
$(\pt[k])_{0\le k\le N}$ associated to each tenor date and appropriate
volatility functions $(\lambda(\cdot,T_k))_{0\le k\le N}$. Let $H$ be a
semimartingale starting from zero, with predictable characteristics $(B,C,\nu)$
or local characteristics $(b,c,F)$ under the terminal measure $\pt$, driving
all LIBOR rates. Then, the dynamics of the forward LIBOR rate with maturity
$T_k$ is
\begin{align}\label{LIBOR-dyn}
L(t,T_k) = L(0,T_k)
 \exp\left(\int_0^t b(s,T_{k})\ud s + \int_0^t\lambda(s,T_k)\ud
H^{T_{k+1}}_s\right)>0,
\end{align}
where $H^{T_{k+1}}$ denotes the \textit{martingale part} of the \smmg $H$ under
the measure $\pt[k+1]$, and the drift term is
\begin{align}\label{LIBOR-drift-term}
b(s,T_{k})
 & = -\half\lambda(s,T_k)^2 c_s\nonumber\\
 & \quad\;
     - \int_{\R} \big(\e^{\lambda(s,T_k)x}-1-\lambda(s,T_k)x\big)
F_s^{T_{k+1}}(\ud x),
\end{align}
ensuring that $L(\cdot,T_k)\in\mathcal M(\P_{T_{k+1}})$. The semimartingale
$H^{T_{k+1}}$ has the $\P_{T_{k+1}}$-canonical decomposition
\begin{align}\label{LIBOR-Levy}
H_t^{T_{k+1}}
      = \int_0^t \sqrt{c_s}\ud W_s^{T_{k+1}}
      + \int_0^t\int_{\R}x(\mu^H-\nu^{T_{k+1}})\dsdx,
\end{align}
where the $\P_{T_{k+1}}$-Brownian motion is
\begin{align}\label{Levy-LIBOR-Brownian}
W_t^{T_{k+1}}
 & = W_t - \int_0^t \left(\sum_{l=k+1}^{N}
     \frac{\delta_l L(t-,T_l)}{1+\delta_l
L(t-,T_l)}\lambda(t,T_l)\right)\sqrt{c_s} \ud s,
\end{align}
and the $\P_{T_{k+1}}$-compensator of $\mu^H$ is
\begin{align}\label{Levy-LIBOR-compensator}
\nu^{T_{k+1}}\dsdx
 &= \left(\prod_{l=k+1}^{N}\frac{\delta_l L(t-,T_l)}{1+\delta_l L(t-,T_l)}
    \Big(\e^{\lambda(t,T_l)x}-1\Big) +1\right)\nu \dsdx.
\end{align}
As an example, the classical log-normal LIBOR model is described in this context
by setting $(b,c,F)=(0,\sigma^2,0)$.

Now, let us discuss some consequences of this modeling approach. Clearly $H$
remains a \smmg under any forward measure, since the class of semimartingales is
closed under equivalent measure changes. However, any \textit{additional}
structure that we impose on the process $H$ to make the model analytically
tractable will be \textit{destroyed} by the measure changes from the terminal to
the forward measures, as the random, state-dependent terms
$\frac{\delta_l L(t-,T_l)}{1+\delta_l L(t-,T_l)}$ entering into eqs.
\eqref{Levy-LIBOR-Brownian} and \eqref{Levy-LIBOR-compensator} clearly indicate.
For example, if $H$ is a \lev process under $\pt$, then $H^{T_{k+1}}$ is not a
\lev process (not even a process with independent increments) under $\pt[k+1]$.
Hence, we have the following consequences:
\begin{enumerate}
\item if $H$ is a \textit{continuous} \smmg, then caplets can be priced in
      closed form, but \textit{not} swaptions or other multi-LIBOR derivatives;
\item if $H$ is a \textit{general} \smmg, then even caplets \textit{cannot} be
      priced in closed form.
\end{enumerate}

Moreover, the Monte Carlo simulation of LIBOR rates in this model is
computationally \textit{very expensive}, due to the complexity evident in eqs.
\eqref{Levy-LIBOR-Brownian} and \eqref{Levy-LIBOR-compensator}. Expressing the
dynamics of the \lib rate in \eqref{LIBOR-dyn} under the terminal measure leads
to a \textit{random drift term}, hence we need to simulate the whole path and
not just the terminal random variable. More severely, the random drift term of
e.g. $L(\cdot,T_k)$ depends on \textit{all} subsequent \lib rates
$L(\cdot,T_l)$, $k+1\le l\le N$. Indeed, the dependence of each rate on all
rates with later maturity can be represented as a strictly (lower) triangular
matrix. Hence, all \lib rates need to be evolved simultaneously in the Monte
Carlo simulation.

Of course, some remedies for the analytical \textit{in}tractability of the \lib
market model have been proposed in the literature; cf.  \citeANP{JoshiStacey08}
\citeyear{JoshiStacey08}
and \shortciteN{GatarekBachertMaksymiuk06} for excellent overviews, focused on
the log-normal LMM. The common practice is to replace the random terms in
\eqref{Levy-LIBOR-Brownian} and \eqref{Levy-LIBOR-compensator} by their
\textit{deterministic} initial values, i.e. to approximate
\begin{align}\label{frozen}
\frac{\delta_l L(t-,T_l)}{1+\delta_l L(t-,T_l)}
 \approx  \frac{\delta_l L(0,T_l)}{1+\delta_l L(0,T_l)};
\end{align}
this is usually called the \textit{``frozen drift'' approximation}. As a
consequence, the structure of the process $H$ will be -- loosely speaking --
preserved under the successive measure changes. For example, if $H$ is a \lev
process, then $H^{T_{k+1}}$ will become a time-inhomogeneous \lev process (due
to the time-dependent volatility function). Hence, caps and swaptions can be
priced in closed form. However, empirical results show that this approximation
does not yield acceptable results.

More recently, \citeN{SiopachaTeichmann07} and Papapantoleon and Siopacha
\citeyear{PapapantoleonSiopacha09} have developed Taylor approximation schemes
for the random terms entering \eqref{Levy-LIBOR-Brownian} and
\eqref{Levy-LIBOR-compensator} using perturbation-based techniques. This method
offers approximations that are more precise than the ``frozen drift''
approximation \eqref{frozen}, while at the same time being faster than
simulating the actual dynamics. Moreover, they offer a theoretical justification
for the ``frozen drift'' approximation as the zero-order Taylor expansion. In
related work, \shortciteN{PapapantoleonSchoenmakersSkovmand10} have developed
log-\lev approximations for the \lev \lib model, thus allowing for accurare very
long stepping in the Monte Carlo simulation.

Concluding, LIBOR market models satisfy Axioms 1 and 2. As far as Axiom 3 is
concerned, LIBOR rates are analytically tractable only under their own forward
measure and only if the driving process is continuous. LIBOR rates are \emph{not
tractable with respect to any other forward measure}. Therefore caps can
(possibly) be priced in closed form, but not swaptions or more exotic
multi-LIBOR derivatives.

\begin{remark}
Additionally, it is econometrically not desirable to model LIBOR rates as
exponentials of processes with independent increments. However we admit that
this is a minor point.
\end{remark}
\end{approach}

\pagebreak

\begin{approach}
In the \emph{forward price model} proposed by Eberlein and \"Ozkan
\citeyear{EberleinOezkan05} and
\citeN{Kluge05}, the forward price -- instead of the \lib rate -- is modeled as
an exponential \lev process or semimartingale. Consider a setting similar to the
previous approach: $(T_k)_{0\le k\le N}$ is a discrete tenor of dates,
$(\pt[k])_{0\le k\le N}$ are forward measures and $H$ denotes a \smmg with
characteristics $(B,C,\nu)$ under the terminal measure $\pt$, where $H_0=0$.
Then, the dynamics of the forward price $F(\cdot,T_k,\tk[k+1])$, or equivalently
of $1+\delta L(\cdot,T_k)$, is given by
\begin{align}\label{FP-dyn}
1+\delta L(t,T_k)
 = (1+\delta L(0,T_k))
   \exp\left(\int_0^t b(s,T_{k})\ud s+\int_0^t \lambda(s,T_k)\ud
H^{T_{k+1}}_s\right),
\end{align}
where $H^{T_{k+1}}$ denotes the martingale part of $H$ under the measure
$\pt[k+1]$ and the drift term $b(s,T_{k})$ is analogous to
\eqref{LIBOR-drift-term}, ensuring that $L(\cdot,T_k)\in\mathcal
M(\P_{T_{k+1}})$. The \smmg $H^{\tk[k+1]}$ has the $\P_{T_{k+1}}$-canonical
decomposition
\begin{align}
H_t^{T_{k+1}}
      = \int_0^t \sqrt{c_s}\ud W_s^{T_{k+1}}
      + \int_0^t\int_{\R}x(\mu^H-\nu^{T_{k+1}})\dsdx,
\end{align}
where the $\P_{T_{k+1}}$-Brownian motion is
\begin{align}\label{Levy-FP-Brownian}
W_t^{T_{k+1}}
 & = W_t - \int_0^t \left(\sum_{l=k+1}^{N}\lambda(t,T_l)\right)\sqrt{c_s} \ud s,
\end{align}
and the $\P_{T_{k+1}}$-compensator of $\mu^H$ is
\begin{align}\label{Levy-FP-compensator}
\nu^{T_{k+1}}\dsdx
 &= \exp\left(x\sum_{l=k+1}^{N}\lambda(s,T_l)\right)\nu\dsdx.
\end{align}

Now, we can immediately deduce from \eqref{Levy-FP-Brownian} and
\eqref{Levy-FP-compensator} that the structure of the process $H$ under $\pt$ is
preserved under any forward measure $\pt[k+1]$. For example, if $H$ is a \lev
process under $\pt$ then it becomes a time-inhomogeneous \lev process under
$\pt[k+1]$, if the volatility function is time-dependent. We can also deduce
that the measure change from the terminal to any forward measure is an Esscher
transformation (cf. \citeNP{KallsenShiryaev02}).

As a result, the model is analytically tractable, thus caps and swaptions can be
priced in semi-analytic form (similarly to an HJM model). Even some
path-dependent derivatives can be priced easily, cf. Kluge and Papapantoleon
\citeyear{KlugePapapantoleon06}. However, \textit{negative} \lib rates can occur
in this model, since forward prices are positive but not necessarily greater
than one. Thus Axiom 1 is \textit{violated}, while Axioms 2 and 3 are satisfied
in the best possible way: \lib rate processes are analytically tractable with
respect to all possible forward measures.
\end{approach}

\begin{remark}\label{HJM-F-eq}
The forward price model can be embedded in the HJM framework with a
\textit{deterministic} volatility structure; cf. \citeN[\S 3.1.1.]{Kluge05}.
\end{remark}

\begin{approach}
\emph{Markov-functional models} were introduced in the seminal paper of
\citeN{HuntKennedyPelsser00}. In contrast to the previous two approaches, the
aim of Markov-functional models is not to model some fundamental quantity, for
example the \lib or swap rate, directly. Instead, Markov-functional models are
constructed by inferring the model dynamics, as well as their functional forms,
through matching the model prices to the market prices of certain liquid
derivatives. That is, they are \emph{implied interest rate models}, and should
be thought of in a fashion similar to local volatility models and implied trees
in equity markets.

The main idea behind Markov-functional models is that bond prices and the
numeraire are, at any point in time, a function of a \emph{low-dimensional}
Markov process under some martingale measure. The functional form for the bond
prices is selected such that the model accurately calibrates to the relevant
market prices, while the freedom to choose the Markov process makes the model
realistic and tractable. Moreover, the functional form for the numeraire can be
used to reproduce the marginal laws of swap rates or other relevant instruments
for the calibration. For further details and concrete applications we refer the
reader to the books by \citeANP{HuntKennedy04} \citeyear{HuntKennedy04} and
\citeN{Fries07}, and the
references therein.

\begin{remark}
One can show that forward price models and affine LIBOR models, that will be
introduced in section \ref{new_approach}, belong to the class of
Markov-functional models, while LIBOR market models do not. In LMMs the \lib
rates are functions of a \textit{high-dimensional} Markov process. Moreover, it
is interesting to compare the properties of a ``good pricing model'' in Hunt et
al. \citeyear[pp. 392]{HuntKennedyPelsser00} with Axioms 1--3.
\end{remark}
\end{approach}

The first two modeling approaches we have reviewed might appear similar at first
sight, but they actually differ in quite fundamental ways -- apart from the
considerations regarding Axioms 1, 2 and 3.

On the one hand, the distributional properties are markedly different. In the
LIBOR market model -- driven by Brownian motion -- LIBOR rates are
\textit{log-normally} distributed, while in the forward price model -- again
driven by Brownian motion -- LIBOR rates are, approximately, \textit{normally}
distributed. Although there seems to be no consensus among market participants
on which assumption is better, it is worth pointing out that in the CEV model --
where for $\beta\to0$ the law is normal and for $\beta\to1$ the law is
log-normal -- a typical value for market data is $\beta\approx0.4$.

On the other hand, changes in the driving process affect \lib rates in the \lib
model and the forward price model in a very different way; see also \citeN[pp.
60]{Kluge05}. Assume that in a small time interval of length $\dt$ the driving
process changes its value by a small amount $\Delta$. Then, in the \textit{LIBOR
market model} we get:
\begin{align}
L(t+\dt,T)
 &= L(t,T) + \Delta\cdot L(t,T) + \mathcal{O}(\Delta^2),
\end{align}
while in the \textit{forward price model} we get:
\begin{align}
L(t+\dt,T) &= L(t,T) + \frac{\Delta}{\delta}
                    + \Delta\cdot L(t,T) + \mathcal{O}(\Delta^2).
\end{align}
Hence, in LMMs changes in the driving process affect the rate roughly
proportional to the current level of the LIBOR rate. In the forward price model
changes do not depend on the actual level of the LIBOR rate. \medskip

\noindent\textbf{Aim:}
\textit{We would like to construct a ``forward price''-type model with positive
\lib rates, i.e. we want a model that respects simultaneously Axioms 1, 2 and
3.}\medskip

A first idea would be to search for a process that makes the martingale in
\eqref{FP-dyn} greater than one, hence guaranteeing that LIBOR rates are always
\textit{positive}. However, such an attempt is doomed to fail since one demands
that
$$
\int_0^t b(s,T_{k})\ud s+\int_0^t \lambda(s,T_k)\ud H^{T_{k+1}}_s \geq 0
$$
with respect to the forward measure (or any other equivalent measure). This
reduces the class of available semimartingales considerably and restricts the
applicability of the models. We show in Section \ref{MMge1} an alternative
construction with rich stochastic structure.

\section{Affine processes}
\label{affine}

Let $(\Omega,\F,\bF,\P)$ denote a complete stochastic basis, where
$\bF=(\F_t)_{t\in[0,T]}$, and let $0 < T \le \infty$ denote some, possibly
infinite, time horizon. We consider a process $X$ of the following type:
\begin{asa}\label{assumption-affine}
Let \prozess[X] be a conservative, time-homogene\-ous, stochastically continuous
Markov process taking values in $D=\Rp^d$, and $(\P_x)_{x\in D}$ a family of
probability measures on $(\Omega,\F)$, such that $X_0 = x$ $\P_x$-almost surely,
for every $x \in D$. Setting
\begin{align}
\mathcal{I}_T
:= \set{u\in\R^d: \E_x\big[\e^{\scal{u}{X_T}}\big] < \infty,
        \,\,\text{for all}\; x \in D},
\end{align}
we assume that
\begin{itemize}
\item[(i)] $0 \in \I_T^\circ$;
\item[(ii)] the conditional moment generating function of $X_t$
  under $\P_x$ has exponentially-affine dependence on $x$; that
  is, there exist functions $\phi_t(u):[0,T]\times\I_T\to\R$ and
  $\psi_t(u):[0,T]\times\I_T\to\R^d$ such that
\begin{align}\label{affine-def}
\E_x\big[\exp\langle u,X_t\rangle\big]
 = \exp\big( \phi_t(u) + \langle\psi_t(u),x\rangle \big),
\end{align}
for all $(t,u,x) \in [0,T] \times \I_T \times D$.
\end{itemize}
\end{asa}
\noindent Here ``$\cdot$'' or $\langle\cdot,\cdot\rangle$ denote the inner
product on $\R^d$, and $\E_x$ the expectation with respect to $\P_x$.

Stochastic processes on $\Rp^d$ with the ``affine property'' \eqref{affine-def}
have been studied since the seventies as the continuous-time limits of
Galton--Watson branching processes with immigration, cf.
\citeN{KawazuWatanabe71}. More recently, such processes on the more general
state space $\Rp^m\times\R^{n}$ have been studied comprehensively, and with a
view towards applications in finance, by \citeN{DuffieFilipovicSchachermayer03}.
We will largely follow their approach, complemented by some results from
\citeN{KellerRessel08}.

By Theorem 3.18 in \citeN{KellerRessel08}, the right hand derivatives
\begin{align}\label{derivatives-FR}
F(u) := \frac{\partial}{\partial t}\big|_{t=0+}\phi_t(u)
 \qquad\text{ and } \qquad
R(u) := \frac{\partial}{\partial t}\big|_{t=0+}\psi_t(u)
\end{align}
exist for all $u\in \I_T$ and are continuous in $u$, such that $X$ is a `regular
affine process' in the sense of \shortciteN{DuffieFilipovicSchachermayer03}.
Moreover, $F$ and $R$ satisfy \lk-type equations. It holds that
\begin{align}\label{F-def}
F(u) = \langle b,u\rangle +
     \int_D\big(\e^{\langle\xi,u\rangle}-1\rangle\big)m(\ud \xi)
\end{align}
and
\begin{align}\label{R-def}
R_i(u) = \langle \beta_i,u\rangle
       + \Big\langle\frac{\alpha_i}2u,u\Big\rangle
       + \int_D\big(\e^{\langle\xi,u\rangle}-1-\langle
u,h^i(\xi)\rangle\big)\mu_i(\ud \xi),
\end{align}
where $(b,m,\alpha_i,\beta_i,\mu_i)_{1\le i\le d}$ are \textit{admissible
parameters}, and $h^i:\Rp^d\to\R^d$ are suitable truncation functions, defined
coordinate-wise by
\begin{align}
h^i_k(\xi) &:=
\begin{cases}
 0,\quad           & k\neq i\\
 \chi(\xi_k),\quad & k= i
\end{cases}
\qquad \text{for all $\xi\in\Rp^d, i\in\set{1,\dots,d}$}\;,
\end{align}
with $\chi(z)$ any bounded Borel function that behaves like $z$ in a
neighborhood of $0$, such as $\frac{z}{1+z^2}$ or $z1_{\set{|z|\le1}}$.

The parameters $(b,m,\alpha_i,\beta_i,\mu_i)_{1\le i\le d}$ have the following
form: $(\beta_i)_\idd$ and $b$ are $\R^d$-valued vectors, $(\alpha_i)_\idd$ are
positive semidefinite real $d\times d$ matrices, and $m$ and $(\mu_i)_\idd$ are
\lev measures on $\Rp^d$. They satisfy additional \emph{admissibility
conditions}; writing $I = \set{1, \dotsc, d}$, these conditions are given,
according to \shortciteN{DuffieFilipovicSchachermayer03}, by
\begin{align}\label{admissible-start}
 b\in \Rp^d
\end{align}
\begin{align}
\beta_{i(k)}\in\Rp \;\, \forall \, k\in I\backslash\{i\}
 \quad\text{ and }\quad
\beta_{i(i)}\in\R
\end{align}
\begin{align}
\alpha_{i(kl)} = 0
 \quad\text{if } k\in I\backslash\{i\}
 \text{ or } l\in I\backslash\{i\}
\end{align}
\begin{align}
 m(\{0\}) = 0 \quad\text{ and }\quad
\int_D (|\xi|\wedge1)m(\dd[\xi]) < \infty
\end{align}
where $|\xi|=\sum_{i}|\xi_i|$ for $\xi\in\R^d$; and, for all $i\in I$
\begin{align}\label{admissible-end}
 \mu_i(\{0\}) = 0 \quad\text{ and }\quad
\int_D \big[ (|\xi_{I\backslash\{i\}}|+|\xi_i|^2)\wedge1
\big]\mu_i(\dd[\xi]) < \infty.
\end{align}

The time-homogeneous Markov property of $X$ implies the following conditional
version of \eqref{affine-def}:
\begin{align}\label{affine-conditional}
\E_x\big[\exp\langle u,X_{t + s}\rangle\big|\F_s\big]
 = \exp\big( \phi_{t}(u) + \langle\psi_{t}(u),X_s\rangle \big),
\end{align}
for all $0 \le t+s \le T$ and $u \in \I_T$. Applying this equation iteratively,
we can deduce that the functions $\phi$ and $\psi$ satisfy the
\textit{semi-flow property}
\begin{equation}\label{flow}
\begin{split}
\phi_{t+s}(u) &= \phi_{t}(u)+\phi_{s}(\psi_t(u))\\
\psi_{t+s}(u) &= \psi_{s}(\psi_{t}(u))
\end{split}
\end{equation}
for all $0 \le t+s \le T$ and $u \in \I_T$, with initial condition
\begin{align}\label{phi-psi-0}
 \phi_0(u)=0
  \quad\text{ and }\quad
 \psi_0(u)=u.
\end{align}
For details we refer to Lemma 3.1 in \shortciteN{DuffieFilipovicSchachermayer03}
and Proposition 1.3 in \citeN{KellerRessel08}.

Differentiating the flow equations (and using the existence of
\eqref{derivatives-FR}) we arrive at the following ODEs (the \textit{generalized
Riccati equations}) satisfied by $\phi_t$ and $\psi_t$:
\begin{subequations}\label{Riccati}
\begin{align}
\frac{\partial}{\partial t}\phi_t(u)
 &= F(\psi_t(u)),  \qquad \phi_0(u)=0, \label{Ric-1}\\
\frac{\partial}{\partial t}\psi_t(u)
 &= R(\psi_t(u)),  \qquad \psi_0(u)=u, \label{Ric-2}
\end{align}
\end{subequations}
for $(t,u)\in [0,T] \times \I_T$; cf. \shortciteN[Theorem
2.7]{DuffieFilipovicSchachermayer03}. If $\psi_t(u)$ stays in $\I_T^\circ$ for
all $t \in [0,T]$, it is a unique solution. Note that if the jump measures $m$
and $\mu$ are zero, then $F(u)$ and each $R_i(u)$ are quadratic polynomials,
whence the differential equations degenerate into classical Riccati equations.

Finally, let us mention, that any choice of admissible parameters satisfying
\eqref{admissible-start}-\eqref{admissible-end}, and corresponding functions $F$
and $R$, gives rise to a uniquely defined affine process, whose moment
generating function can be calculated through the generalized Riccati equations
\eqref{Riccati}.

\begin{remark}\label{affine-1d}
We mention here the following examples of one-dimensional processes satisfying
Assumption~$(\mathbf{A})$:
\begin{enumerate}
\item Every \lev subordinator with cumulant generating function
 $\kappa(u)$ and finite exponential moment; it is characterized by
 the functions $F(u)=\kappa(u)$ and $R(u)=0$.
\item Every OU-type process (cf. \citeNP[section 17]{Sato99}) driven
 by a \lev subordinator with finite exponential moment; such a process
 is characterized by $F(u)=\kappa(u)$ and $R(u)=\beta u$, with
 $\beta\in\R$.
\item The squared Bessel process of dimension $\alpha$ (cf.
 \citeNP[Ch. XI]{RevuzYor99}), characterized by $F(u)=\alpha u$
 and $R(u)=2u^2$, with $\alpha>0$.
\end{enumerate}
\end{remark}

Finally, we will later need the following results; let us denote by
$(\ecc)_{\icc\le d}$ the unit vectors in $\R^d$ and let inequalities involving
vectors be interpreted component-wise.

\begin{lemma}\label{positivity}
The functions $\phi$ and $\psi$ satisfy the following:
\begin{enumerate}
\item $\phi_t(0) = \psi_t(0) = 0$ for all $t \in [0,T]$.
\item $\I_T$ is a convex set; moreover, for each $t \in [0,T]$,
 the functions $\I_T \ni u \mapsto \phi_t(u)$ and
 $\I_T \ni u \mapsto \psi_t(u)$ are (componentwise) convex.
\item $\phi_t(\cdot)$ and $\psi_t(\cdot)$ are order-preserving: let
 $(t,u),(t,v)\in [0,T] \times \I_T$, with $u\le v$.
 Then
\begin{align}\label{order}
\phi_t(u)\le \phi_t(v)
 \quad\text{ and }\quad
\psi_t(u)\le \psi_t(v).
\end{align}
\item $\psi_t(\cdot)$ is \emph{strictly} order-preserving: let
 $(t,u),(t,v) \in [0,T] \times \I_T^\circ$, with $u < v$. Then
 $\psi_t(u) < \psi_t(v)$.
\end{enumerate}
\end{lemma}
\begin{proof}
From \eqref{F-def} and \eqref{R-def} it is immediately seen that $F(0)=R(0)=0$.
Thus $\phi_t(0) = \psi_t(0) = 0$ are solutions to the corresponding generalized
Riccati equations \eqref{Riccati}. Moreover, $0 \in \I_T^\circ$, such that the
solutions are unique, showing claim (1). Let $u, v \in \R^d$ and
$\lambda \in [0,1]$. By H\"o{}lder's inequality
\begin{align}
\E_x\big[\exp\left(\scal{\lambda u + (1 - \lambda) v}{X_t}\right)\big]
 \le \E_x\big[\e^{\scal{u}{X_t}}\big]^\lambda
     \cdot \E_x\big[\e^{\scal{v}{X_t}}\big]^{(1 - \lambda)}\;,
\end{align}
where both sides may take the value $+\infty$. Taking logarithms on both sides
shows that, for all $t \in [0,T]$, $\phi_t(\cdot)$ and $\psi_t(\cdot)$ are
(componentwise) convex functions on $\R^d$, taking values in the extended real
numbers $\R \cup \set{+\infty}$. This implies in particular that $\I_T$ is
convex, and that the restrictions of $\phi_t(\cdot)$ and $\psi_t(\cdot)$ to
$\I_T$ are finite convex functions, showing claim (2). Following
\citeN[Proposition 1.3(vii)]{KellerRessel08}, we have that for $u\le v$
\begin{align*}
\E_x\big[\e^{\langle u,X_t\rangle}\big]
 \le \E_x\big[\e^{\langle v,X_t\rangle}\big]
 < \infty,
\end{align*}
for all $x\in \Rp^d$. Now, using the affine property of the moment generating
function we get
\begin{align}
\phi_t(u) + \langle\psi_t(u),x\rangle
\le \phi_t(v) + \langle\psi_t(v),x\rangle,
\end{align}
whereby inserting first $x=0$ and then $x=C\ecc$, for $C>0$ arbitrarily large,
yields claim (3). Consider the Riccati differential equation \eqref{Ric-2},
satisfied by $\psi_t$. By \citeN{KellerRessel08}, Lemma~4.6, $R(u)$ is
quasi-monotone increasing; moreover, it is locally Lipschitz in $\I_T^\circ$. A
comparison principle for quasi-monotone ODEs (cf.
\citeNP{Walter96},~Section~10.XII) yields then directly that $u < v$ implies
$\psi_t(u) < \psi_t(v)$ for all $t \in [0,T]$.
\end{proof}

The above results on affine processes can be extended to the case when the
time-homogeneity assumption on the Markov process \prozess[X] is dropped, see
\citeN{Filipovic05}. The conditional moment generating function then takes the
form
\begin{align}\label{affine-conditional-inhomogeneous}
\E_x\left[\left.\exp\langle u,X_{r}\rangle\right|\F_s \right]
 = \exp\big( \phi_{s,r}(u) + \langle\psi_{s,r}(u),X_s\rangle \big),
\end{align}
for all $(s,r,u)$ such that $0 \le s \le r \le T$ and $u \in \I_T$, with
$\phi_{s,r}(u)$ and $\psi_{s,r}(u)$ now depending on both $s$ and $r$. Assuming
that $X$ satisfies the `strong regularity condition' (cf.~\citeNP{Filipovic05},
Definition~2.9), $\phi_{s,r}(u)$ and $\psi_{s,r}(u)$ satisfy generalized Riccati
equations with time-dependent right-hand sides:
\begin{align}
- \frac{\partial}{\partial s}\phi_{s,r}(u)
 &= F(s,\psi_{s,r}(u)),  \qquad \phi_{r,r}(u)=0, \label{Ric-TI-1}\\
- \frac{\partial}{\partial s}\psi_{s,r}(u)
 &= R(s,\psi_{s,r}(u)),  \qquad \psi_{r,r}(u)=u, \label{Ric-TI-2}
\end{align}
for all $0 \le s \le r \le T$ and $u \in \I_T$.

\section{Constructing Martingales $\geqslant1$}
\label{MMge1}

In this section we construct martingales that stay greater than one for all
times, up to a \emph{bounded} time horizon $T$, that is, from now on
$0<T<\infty$. The construction is a ``backward'' one, and utilizes the Markov
property of affine processes.

\begin{theorem}\label{M-is-martingale}
Let $X$ be an affine process satisfying Assumption $(\mathbf{A})$, and let $u
\in \I_T$. The process \prozess[M^u] defined by
\begin{align}\label{mg-1}
 M_t^u = \exp\big(\phi_{T-t}(u) + \langle\psi_{T-t}(u),X_t\rangle\big),
\end{align}
is a martingale. Moreover, if $u\in \I_T\cap\Rp^d$ then $M^u_t\geq1$ a.s. for
all $t\in[0,T]$, for any $X_0\in\Rp^d$.
\end{theorem}
\begin{proof}
First, we show that $M^u$ is a martingale; for all $u \in \I_T$ it holds that
\begin{align*}
 \E_x[M^u_T] = \E_x[\e^{\langle u,X_T\rangle}] <\infty.
\end{align*}
Moreover, using \eqref{phi-psi-0} and \eqref{affine-conditional}, we have that:
\begin{align*}
\E_x\big[M_T^u|\F_t\big]
 &= \E_x\big[\exp\langle u,X_T\rangle|\F_t\big] \\
 &= \exp\big(\phi_{T-t}(u) + \langle\psi_{T-t}(u),X_t\rangle\big)
  = M_t^u.
\end{align*}
Regarding the assertion that $M_t^u \ge 1$ for all $t\in[0,T]$, it suffices to
note that if $u\in \I_T\cap\Rp^d$, then $M_t^u$ is the conditional expectation
of a random variable greater than, or equal to, one, i.e.
\begin{align}
M_t^u = \E_x\big[\exp\langle u,X_T\rangle\big|\F_t\big] \ge1,
\end{align}
hence greater than, or equal to, one itself.
\end{proof}

\begin{remark}
Actually, the same construction would create martingales for any \emph{Markov}
process on a general state space. Indeed, let $X$ be a Markov process with state
space $\R^e\times\Rp^d$ and consider the random variable
$Y_T^u=\e^{\langle u,X_T\rangle}$. The tower property of conditional
expectations yields that the process \prozess[M^u] with
\begin{align}
 M^u_t = \E\big[Y^u_T|\F_t\big] = \E\big[\e^{\langle u,X_T\rangle}|\F_t\big]
\end{align}
is a martingale. However, taking the positive orthant as state space guarantees
that the martingales stay greater than one. In addition, taking an affine
process as the driving motion provides the appropriate trade-off between rich
structural properties and analytical tractability.
\end{remark}

\begin{example}[\lev processes]
Assume that the affine process $X$ is actually a \textit{\lev subordinator},
with cumulant generating function $\kappa$. Then, we know that
\begin{align}
\phi_t(u)=t\cdot\kappa(u)
\quad \text{ and } \quad
\psi_t(u)=u.
\end{align}
Hence, the exponential martingale in \eqref{mg-1} takes the form:
\begin{align}\label{mg-lev}
M_t^u
 &= \exp\big( \phi_{T-t}(u) + \langle\psi_{T-t}(u),X_t\rangle \big) \nonumber\\
 &= \exp\big( (T-t)\kappa(u) + \langle u,X_t\rangle \big),
\end{align}
which is a martingale by standard results for \lev processes. Moreover, for
$u\in\I_T$, since $\kappa:\I_T\to\Rp$ and $T-t\ge0$, we get that $M_t^u\ge1$ for
all $t\in[0,T]$.
\end{example}

\begin{remark}
These considerations show that the affine \lib model will contain the
\lev forward price model of \citeN{EberleinOezkan05} and \citeN{Kluge05} as a
special case, if we consider a time-inhomogeneous affine process with state
space $\R^d$ as driving motion. Of course, in that case the martingales $M^u$
will \emph{not} be greater than one.
\end{remark}

Note that there is still some ambiguity lurking in the specification of the
martingale $M^u$: consider a $d$-dimensional driving process $X$, from which the
martingale $M^u$ is constructed. Let $c$ be a positive semidefinite $d \times d$
matrix, and $c'$ its transpose. Define $\widetilde{X} = c \cdot X$ and let
$\widetilde{M}^u$ be the corresponding martingale. It is easy to check that if
$X$ is an affine process satisfying condition $\mathbf{A}$, then so is
$\widetilde{X}$. It holds that
\begin{equation*}
M_t^{c'u}
 = \E_x\big[\exp\langle c'u,X_T\rangle\big|\F_t\big]
 = \E_x\big[\exp\langle u,c X_T\rangle\big|\F_t\big]
 = \widetilde{M}^u_t,
\end{equation*}
showing that in terms of the martingales $M^{u}$, a (positive) linear
transformation $c$ of the underlying process $X$ is simply equivalent to the
transposed linear transformation $c'$ of the parameter $u$. In order to avoid
this ambiguity in the specification of the martingale $M^u$, we will fix from
now on the initial value of the process $X$ at some \textit{strictly positive},
\emph{canonical} value, e.g. $\mathbf{1} = (1,\dots,1)$.

Finally, the following definition will be needed later.
\begin{definition}\label{def-gamma}
For any process \prozess[X] satisfying Assumption $(\mathbf{A})$, define
\begin{align}
\gamma_X
 = \sup_{u\in\I_T\cap\R^d_{>0}} \E_{\mathbf{1}}\big[\e^{\scal{u}{X_T}}\big].
\end{align}
\end{definition}

\section{The affine LIBOR models}
\label{new_approach}

Now, we describe our proposed approach to modeling LIBOR rates that aims at
combining the advantages of both the LIBOR and the forward price approach; that
is, a framework that produces \textit{non-negative} LIBOR rates in an
\textit{analytically tractable} model.

Consider a discrete tenor $0=T_0<T_1<T_2<\cdots<T_N = T$ and an initial tenor
structure of \emph{non-negative} LIBOR rates $L(0,T_k)$, $k\in\set{1,\dots,N}$.
We have that discounted traded assets (bonds) are martingales with respect to
the terminal martingale measure, i.e.
\begin{align}
 \frac{B(\cdot,T_k)}{B(\cdot,T_N)} \in \mathcal{M}(\P_{T_N}),
 \qquad\text{for all}\; k\in\{1,\dots,N-1\}.
\end{align}
In the \textit{affine \lib model} we model quotients of bond prices using the
martingales $M^u$ defined in Theorem~\ref{M-is-martingale} as follows:
\begin{subequations}\label{ALM}
\begin{align}
\frac{B(t,T_1)}{B(t,T_N)} &= M_t^{u_1} \label{mg-2}\\
 &\vdots\nonumber\\
\frac{B(t,T_{N-1})}{B(t,T_N)} &= M_t^{u_{N-1}}, \label{mg-N}
\end{align}
\end{subequations}
for all $t\in[0,T_1],\dots,t\in[0,T_{N-1}]$ respectively. The initial values of
the martingales $M^{u_k}$ must satisfy:
\begin{align}\label{initial-cond}
M^{u_k}_0
 = \exp\big(\phi_T(u_k) + \big\langle\psi_T(u_k),x\big\rangle\big)
 = \frac{B(0,T_k)}{B(0,T_N)}
\end{align}
for all $k\in\{1,\dots,N-1\}$. Obviously we set
$u_N=0\Leftrightarrow M_0^{u_N}=\frac{B(0,T_N)}{B(0,T_N)}=1$.

Next, we show that under mild conditions on the underlying process $X$, an
affine LIBOR model can fit \textit{any} given term structure of initial LIBOR
rates through the parameters $u_1,\dots,u_N$.

\begin{proposition}\label{initial-fit}
Suppose that $L(0,T_1), \dotsc, L(0,T_N)$ is a tenor structure of non-negative
initial LIBOR rates, and let $X$ be a process satisfying assumption
$(\mathbf{A})$, starting at the canonical value $\mathbf{1}$. The following
hold:
\begin{enumerate}
\item If $\gamma_X>B(0,T_1)/B(0,T_N)$, then there exists
 a decreasing sequence \linebreak
 $u_1\ge u_2\ge\dots\ge u_N =0$ in $\I_T\cap\Rp^d$, such that
 \begin{equation}\label{martingale-initial}
  M_0^{u_k} = \frac{B(0,T_k)}{B(0,T_N)},
   \qquad \text{for all }\, k\in\set{1,\dotsc,N}.
 \end{equation}
 In particular, if $\gamma_X=\infty$, then the affine LIBOR model
 can fit \emph{any} term structure of non-negative initial LIBOR
 rates.
\item If $X$ is one-dimensional, the sequence
 $(u_k)_{k\in\set{1,\dotsc,N}}$ is \emph{unique}.
\item If all initial LIBOR rates are positive, the sequence
 $(u_k)_{k\in\set{1,\dotsc,N}}$ is \emph{strictly} decreasing.
\end{enumerate}
\end{proposition}
\begin{proof}
The non-negativity of (initial) LIBOR rates clearly implies that
\begin{align*}
\frac{B(0,T_1)}{B(0,T_N)}
 \ge \frac{B(0,T_2)}{B(0,T_N)}
 \ge \dotsm
 \ge \frac{B(0,T_N)}{B(0,T_N)} = 1.
\end{align*}
Moreover, if the initial LIBOR rates are positive the above inequalities become
strict. Now let $\epsilon > 0$, small enough such that
$\gamma_X-\epsilon>\frac{B(0,T_1)}{B(0,T_N)}$. Clearly, by the definition of
$\gamma_X$, we can find some $u_+ \in \I_T 0$ such that
\begin{align*}
\E_{\mathbf{1}}\left[\e^{\scal{u_+}{X_T}}\right]
 > \gamma_X  - \epsilon
 > \frac{B(0,T_1)}{B(0,T_N)}.
\end{align*}
Define now
\begin{equation}\label{auxilliary-function}
f: [0,1] \to \Rp, \quad
  \xi \mapsto \E_{\mathbf{1}}\left[\e^{\scal{\xi u_+}{X_T}}\right]
               = M_0^{\xi u_+}.
\end{equation}
By monotone convergence and dominated convergence, $f$ is a continuous,
increasing function satisfying $f(0)=1$ and $f(1)>\frac{B(0,T_1)}{B(0,T_N)}$.
Consequently, there exist numbers $0=\xi_N\le\xi_{N-1}\le\dotsc\le\xi_1< 1$,
such that
\begin{align*}
f(\xi_k)
 = M_0^{\xi u_+}
 = \frac{B(0,T_k)}{B(0,T_N)},
 \quad \text{for all}\;k \in \set{1,\dotsc,N}.
\end{align*}
Setting $u_k = \xi_k u_+$, we have shown \eqref{martingale-initial}. By
Lemma~\ref{positivity}, $f(\xi)$ is in fact a strictly increasing function. If
also the (quotients of) bond prices satisfy strict inequalities, we deduce that
the sequence $(u_k)_{k\in\set{1,\dots,N}}$ is strictly decreasing, showing claim
(3). Finally, if $X$ is one-dimensional, then $\I_T\cap\Rp$ is just a
sub-interval of the positive half-line; thus any choice of $u_+$, will lead to
the same parameters $u_k$, showing (2).
\end{proof}

In the affine LIBOR model, forward prices have the following dynamics:
\begin{align}\label{mg-3}
\frac{B(t,T_k)}{B(t,\tk[k+1])}
 &= \frac{B(t,T_k)}{B(t,T_N)}\frac{B(t,T_N)}{B(t,\tk[k+1])}
  = \frac{M_t^{u_k}}{M_t^{u_{k+1}}} \nonumber\\
 &= \exp\Big(\phi_{T_N-t}(u_k)- \phi_{T_N-t}(u_{k+1})\nonumber\\
 &\qquad\qquad
  + \big\langle\psi_{T_N-t}(u_k)-\psi_{T_N-t}(u_{k+1}),X_t\big\rangle\Big)
\nonumber\\
 &= \exp\Big(A_{T_N-t}(u_k,u_{k+1})
     + \big\langle B_{T_N-t}(u_k,u_{k+1}),X_t\big\rangle\Big),
\end{align}
where we have defined
\begin{subequations}
\begin{align}
 A_{T_N-t}(u_k,u_{k+1})
 := \phi_{T_N-t}(u_k)- \phi_{T_N-t}(u_{k+1}),\\
 B_{T_N-t}(u_k,u_{k+1})
 := \psi_{T_N-t}(u_k)- \psi_{T_N-t}(u_{k+1}).
\end{align}
\end{subequations}

Using Proposition \ref{initial-fit}(1) and Lemma \ref{positivity}(3), we
immediately deduce the following result, which shows that Axiom 1 is satisfied:

\begin{proposition}\label{libor-positiv}
Suppose that $L(0,T_1), \dotsc, L(0,T_N)$ is a tenor structure of non-negative
initial LIBOR rates, and let $X$ be a process satisfying assumption
$(\mathbf{A})$. Let the bond prices be given by \eqref{ALM} and satisfy the
initial conditions \eqref{initial-cond}. Then the LIBOR rates $L(t,T_k)$ are
\emph{non-negative} a.s., \emph{for all} $t\in[0,T_k]$ and $k \in
\set{1,\dotsc,N}$.
\end{proposition}

Moreover, forward prices should be martingales with respect to their
corresponding forward measures, that is
\begin{align}
 \frac{B(\cdot,T_k)}{B(\cdot,\tk[k+1])}
 \in\mathcal{M}(\P_{T_{k+1}}),
 \qquad\text{for all}\; k\in\{1,\dots,N-1\};
\end{align}
this we can easily deduce in our modeling framework. Forward measures are
related to each other via forward processes, hence in the present framework
forward measures are related to one another via quotients of the martingales
$M^u$. Indeed, we have that
\begin{align}\label{Pk-to-next}
\frac{\dd \P_{\tk[k]}}{\dd \P_{\tk[k+1]}} \Big|_{\F_t}
  = \frac{F(t,\tk[k],\tk[k+1])}{F(0,\tk[k],\tk[k+1])}
  = \frac{B(0,\tk[k+1])}{B(0,\tk[k])} \times \frac{M_t^{u_{k}}}{M_t^{u_{k+1}}}
\end{align}
for any $k\in\{1,\dots,N-1\}, t\in[0,\tk[k]]$. Then, using Proposition III.3.8
in \citeN{JacodShiryaev03} we can easily deduce that $L(\cdot,T_k)$ is a
martingale under the forward measure $\P_{\tk[k+1]}$, since the successive
densities from $\P_{\tk[k+1]}$ to $\P_{\tk[N]}$ yield a ``telescoping'' product
and a $\P_{\tk[N]}$ martingale. We have that
\begin{align}
 1+\delta L(\cdot,T_{k})
 = \frac{B(\cdot,\tk[k])}{B(\cdot,\tk[k+1])}
 = \frac{M^{u_{k}}}{M^{u_{k+1}}}
 \in \mathcal{M}(\P_{\tk[k+1]})
\end{align}
since
\begin{align}
\frac{M^{u_{k}}}{M^{u_{k+1}}} \cdot \frac{\dd \P_{\tk[k+1]}}{\dd \P_{\tk[N]}}
\stackrel{\eqref{Pk-to-next}}{=}
\frac{M^{u_{k}}}{M^{u_{k+1}}}
 \prod_{l=k+1}^{N-1} \frac{M^{u_{l}}}{M^{u_{l+1}}}
 &= M^{u_{k}}
 \in \mathcal{M}(\P_{\tk[N]})
\end{align}
by the construction of the model. Hence Axiom 2 is also satisfied.

In addition, we get that the density between the $\P_{\tk[k]}$-forward measure
and the terminal forward measure $\P_{\tk[N]}$ is given by the martingale
$M^{u_k}$, as the defining equations \eqref{ALM} already dictate; we have
\begin{align}\label{Pk-to-final}
\frac{\dd \P_{\tk[k]}}{\dd \P_{\tk[N]}} \Big|_{\F_t}
  = \frac{B(0,\tk[N])}{B(0,\tk[k])} \times \frac{B(t,\tk[k])}{B(t,\tk[N])}
  = \frac{B(0,\tk[N])}{B(0,\tk[k])} \times M_t^{u_{k}} =
\frac{M_t^{u_k}}{M_0^{u_k}}.
\end{align}
This we can also deduce by expanding the densities between $\P_{\tk[k]}$ and
$\P_{\tk[N]}$.

Next, we wish to show that the model structure is preserved under any forward
measure. Indeed, changing from the terminal to the forward measure $X$ becomes
a time-inhomogeneous Markov process, but the affine property of its moment
generating function is preserved. This means that Axiom 3 is satisfied in full
strength: $X$ will be a \emph{time-inhomogeneous affine} process under \emph{any
forward measure}. In order to show this, we calculate the conditional moment
generating function of $X_r$ under the forward measure $\P_{\tk[k]}$, and get
that
\begin{align}\label{Pk-mgf}
&\E_{\pt[k]}\Big[\e^{\langle v,X_r\rangle}\big|\F_s\Big] \nonumber\\
 &= \E_{\pt}\bigg[\frac{M_{r}^{u_{k}}}{M_s^{u_k}}\e^{\langle
v,X_{r}\rangle}\big|\F_s\bigg] \nonumber\\
 &= \frac{1}{M_s^{u_k}}
    \E_{\pt}\Big[\exp\big(\phi_{T_N-r}(u_{k}) +
\langle\psi_{T_N-r}(u_{k}),X_r\rangle
                          + \langle v,X_r\rangle\big)\big|\F_s\Big] \nonumber\\
 &= \exp\big(- \phi_{T_N-s}(u_k) - \scal{\psi_{T_N-s}(u_k)}{X_s} +
\phi_{T_N-r}(u_{k})\big) \nonumber\\
 &\qquad\times
\E_{\pt}\Big[\exp\big(\langle\psi_{T_N-r}(u_{k})+v,X_r\rangle\big)\big|\F_s\Big]
\nonumber\\
 &= \exp\Big(\phi_{T_N-r}(u_{k}) - \phi_{T_N-s}(u_k) +
\phi_{r-s}(\psi_{T_N-r}(u_{k})+ v)\nonumber\\
 &\qquad\qquad
             + \scal{\psi_{r-s}(\psi_{T_N-r}(u_{k})+ v) -
\psi_{T_N-s}(u_k)}{X_s}\Big)\nonumber\\
 &\stackrel{\eqref{flow}}{=}
   \exp\Big(\phi_{r-s}(\psi_{T_N-r}(u_{k})+ v) -
\phi_{r-s}(\psi_{T_N-r}(u_k))\nonumber\\
 &\qquad\qquad
             + \scal{\psi_{r-s}(\psi_{T_N-r}(u_{k})+ v) -
\psi_{r-s}(\psi_{T_N-r}(u_k))}{X_s}\Big),
\end{align}
which yields the affine property of $X$ under the forward measure $\pt[k]$, for
any $k\in\{1,\dots,N-1\}$. In particular, setting $s=0$, $r = t$, we get that
\begin{equation}\label{Pk-mgf-2}
\E_{\pt[k]}\big[\e^{\langle v,X_{t}\rangle}\big] =
\exp\left(\phi^k_t(v) + \scal{\psi^k_t(v)}{x}\right)\;,
\end{equation}
where
\begin{subequations}\label{Pk-mgf-0}
\begin{align}
\phi^k_t(v) &:= \phi_t(\psi_{T_N-t}(u_{k})+ v) - \phi_t(\psi_{T_N-t}(u_k)),
\label{Pk-mgf-3}\\
\psi^k_t(v) &:= \psi_t(\psi_{T_N-t}(u_{k})+ v) - \psi_t(\psi_{T_N-t}(u_k)),
\label{Pk-mgf-4}
\end{align}
\end{subequations}
showing clearly that the measure change from $\pt[k]$ to $\pt$ is an exponential
tilting (or Esscher transformation). Furthermore, we can calculate from
\eqref{Pk-mgf} the functions $F^k(r,v)$ and $R^k(r,v)$, characterizing the
time-inhomogeneous affine process $X$ under the forward measure $\pt[k]$:
\begin{align}
F^k(r,v)
 &= - \frac{\partial}{\partial s} \phi_{r-s}(\psi_{T- r}(u_k) + v)\big|_{s=r}
    + \frac{\partial}{\partial s} \phi_{r-s}(\psi_{T- r}(u_k)) \big|_{s=r}
\nonumber\\
 &= F\left(\psi_{T-r}(u_k) + v\right) - F\left(\psi_{T-r}(u_k)\right),
\intertext{and}
R^k(r,v)
 &= - \frac{\partial}{\partial s} \psi_{r-s}(\psi_{T- r}(u_k) + v)\big|_{s=r}
    + \frac{\partial}{\partial s} \psi_{r-s}(\psi_{T- r}(u_k)) \big|_{s=r}
\nonumber\\
 &= R\left(\psi_{T-r}(u_k) + v\right) -  R\left(\psi_{T-r}(u_k)\right).
\end{align}

Note that the moment generating function in \eqref{Pk-mgf-2} is well defined for
all $v\in\mathcal{I}^k$, where
\begin{align*}
\mathcal{I}^k
 = \big\{ v\in\R^d:\,\, v+\psi_{T_N-t}(u_k)\in\mathcal{I}_T, \,\, t\in[0,T_k]
\big\}.
\end{align*}

Finally, we would like to calculate the moment generating function for the
dynamics of forward prices under their corresponding forward measures. Let us
use the following shorthand notation for \eqref{mg-3}
\begin{align}\label{simple}
& \frac{M_{T_k}^{u_k}}{M_{T_k}^{u_{k+1}}}
  = \e^{A_k + B_k\cdot X_{T_k}},
\end{align}
where
\begin{equation}
\begin{split}
 A_k & := A_{T_N-T_k}(u_k,u_{k+1})
 = \phi_{T_N-T_k}(u_k)- \phi_{T_N-T_k}(u_{k+1}),\\
 B_k & := B_{T_N-T_k}(u_k,u_{k+1})
 = \psi_{T_N-T_k}(u_k)- \psi_{T_N-T_k}(u_{k+1}),
\end{split}
\end{equation}
for any $k\in\{1,\dots,N-1\}$. Then, using \eqref{Pk-mgf-2} we get that
\begin{align}\label{MGF-FOR-tk}
&\E_{\pt[k+1]}\big[\e^{v(A_k + B_k\cdot X_{T_k})}\big] \nonumber\\
 &= \e^{vA_k}\E_{\pt[k+1]}\big[\e^{\langle vB_k,X_{T_k}\rangle}\big] \nonumber\\
 &= \e^{vA_k}\exp\Big(\phi_{T_k}(\psi_{T_N-T_k}(u_{k+1})+ vB_k) -
                      \phi_{T_k}(\psi_{T_N-T_k}(u_{k+1}))\nonumber\\
 &\qquad\qquad\quad
             + \langle\psi_{T_k}(\psi_{T_N-T_k}(u_{k+1})+ vB_k) -
                 \psi_{T_k}(\psi_{T_N-T_k}(u_{k+1})),x\rangle\Big)\nonumber\\
 &\stackrel{\eqref{flow}}{=}
  \frac{B(0,T_N)}{B(0,T_{k+1})}\times
    \exp\Big( v\phi_{T_N-T_k}(u_k) + (1-v)\phi_{T_N-T_k}(u_{k+1})  \\ \nonumber
 &\qquad\qquad\qquad\qquad
    + \phi_{T_k}\big(v\psi_{T_N-T_k}(u_k) + (1-v)\psi_{T_N-T_k}(u_{k+1})\big)\\
 &\qquad\qquad\qquad\qquad
   + \big\langle\psi_{T_k}\big(v\psi_{T_N-T_k}(u_k) +
   (1-v)\psi_{T_N-T_k}(u_{k+1})\big),x\big\rangle \Big).  \nonumber
\end{align}
Note that the moment generating function is again exponentially-affine in the
initial value $X_0=x$. Here, the moment generating function in
\eqref{MGF-FOR-tk} is well defined for all $v\in\mathcal{J}^k$, where
\begin{align*}
\mathcal{J}^k
 = \big\{ v\in\R:\,\,
  v\psi_{T_N - T_k}(u_k)+(1-v)\psi_{T_N - T_k}(u_{k+1})\in\mathcal{I}_T, \big\}.
\end{align*}

Concluding, we have shown that forward prices are of exponential-affine form
under \emph{any} forward measure and the model structure is always preserved. As
a consequence, the model is \emph{analytically tractable} in the sense of Axiom
3 with respect to all forward measures.

\begin{remark}
Note that for the model to make sense and be easy to use and implement we must
know the functions $\phi$ and $\psi$ \emph{explicitly}, and not only implicitly
as solutions of Riccati ODEs.
\end{remark}

\begin{remark}
A particular feature of affine LIBOR models based on one-dimensional driving
processes is that the LIBOR rate $ L(t,T_k) $ is bounded from below by
$\frac{1}{\delta}(\exp[A_{T_N-t}(u_k,\uk)] -1)$. This is undesirable, but a
negligible failure, since usually the quantity $A_{T_N-t}(u_k,\uk)$ is close to
$0$.
\end{remark}

\section{Interest rate derivatives}
\label{derivatives}

The most liquid interest rate derivatives are caps, floors and swaptions. In
practice LIBOR models are typically calibrated to the implied volatility
surface of caps and at-the-money swaptions, and then hedging strategies and
prices of exotic options are deduced. Thus, it is important to have fast
valuation formulas for these options so that the model can be calibrated in real
time. Here we derive semi-analytical formulas for caps and swaptions, making use
of Fourier transform methods. Closed-form solutions for the CIR driving process
will be derived in the next section. We also briefly discuss hedging issues.

\subsection{Caps and floors}

Caps are series of call options on the successive \lib rates, termed caplets,
while floors are series of put options on LIBOR rates, termed floorlets. Caplets
and floorlets are usually settled \emph{in arrears}, i.e. the caplet with
maturity $T_k$ is settled at time $T_{k+1}=T_k+\delta$. We consider a tenor
structure with constant tenor length $\delta$ for simplicity, although this
assumption can be easily relaxed. A cap has the payoff
\begin{align}
\sum_{k=1}^{N-1} \delta (L(T_k,T_k)-K)^+.
\end{align}
Keeping the basic relationship \eqref{basic} in mind, we can re-write caplets as
call options on forward prices:
\begin{align}\label{last-capl}
\delta(L(T_{k},T_{k}) - K)^+
 &= (1+\delta L(T_{k},T_{k}) - 1-\delta K)^+ \nonumber\\
 &= \Big(\frac{M_{\tk[k]}^{u_{k}}}{M_{\tk[k]}^{u_{k+1}}} - \mathscr{K}\Big)^+,
\end{align}
where $\mathscr{K}:=1+\delta K$.

Each individual caplet is typically priced under its corresponding forward
measure to avoid the evaluation of a joint law or characteristic function. In
our modeling framework we have that
\begin{align}\label{caplet-1}
\mathbb{C}(T_k,K)
 &= B(0,\tk[k+1])\, \E_{\P_{\tk[k+1]}}
    \big[\delta\big(L(T_k,T_k) - K\big)^+\big] \nonumber\\
 &= B(0,\tk[k+1])\, \E_{\P_{\tk[k+1]}}
    \Big[\Big(\frac{M_{\tk[k]}^{u_{k}}}{M_{\tk[k]}^{u_{k+1}}} -
\mathscr{K}\Big)^+\Big].
\end{align}
Then, we can apply Fourier methods to calculate the price of this caplet as an
ordinary call option on the forward price.

\begin{proposition}
The price of a caplet with strike $K$ maturing at time $T_k$ is given by the
formula
\begin{align}\label{caplet-0}
\mathbb{C}(T_k,K)
 &= \frac{B(0,\tk[N])\mathscr{K}}{2\pi} \int_{\R} \mathscr{K}^{-R+\ii v}
    \frac{\Lambda_{A_k+B_k\cdot X_{\tk[k]}}(R-\ii v)}{(R-\ii v)(R-1-\ii v)} \dv,
\end{align}
where $R\in\mathcal{J}^k\cap(1,\infty)$ and the moment generating function
$\Lambda_{A_k+B_k\cdot X_{\tk[k]}}$ is given by \eqref{MGF-FOR-tk} via
\begin{align}
\Lambda_{A_k+B_k\cdot X_{\tk[k]}}(v)
 &= \frac{B(0,\tk[k+1])}{B(0,\tk)}
    \E_{\pt[k+1]}\big[\e^{v(A_k+B_k\cdot X_{\tk[k]})}\big].
\end{align}
\end{proposition}
\begin{proof}
Starting from \eqref{caplet-1} and recalling the notation \eqref{simple}, we get
that
\begin{align}\label{caplet-2}
\mathbb{C}(T_k,K)
 &= B(0,\tk[k+1])\, \E_{\P_{\tk[k+1]}}
    \Big[\Big(\e^{A_k+B_k\cdot X_{T_k}} - \mathscr{K}\Big)^+\Big],
\end{align}
hence we can view this as a call option on the random variable
$A_k+B_k\cdot X_{T_k}$. Now, since the moment generating function
$\Lambda_{A_k+B_k\cdot X_{\tk[k]}}$ of $A_k+B_k\cdot X_{T_k}$ is finite for
$R\in\mathcal{J}^k$, and the dampened payoff function of the call option is
continuous, bounded and integrable, and has an integrable Fourier transform for
$R\in(1,\infty)$, we can apply Theorem 2.2 in
\shortciteN{EberleinGlauPapapantoleon08} and immediately get that
\begin{align*}
\mathbb{C}(T_k,K)
 &= \frac{B(0,\tk[k+1])}{2\pi} \int_\R \mathscr{K}^{1+\ii v-R}
\frac{\E_{\pt[k+1]}\big[\e^{(R-\ii v)(A_k+B_k\cdot X_{\tk[k]})}\big]}
     {(\ii v-R)(1+\ii v-R)}\dv,
\end{align*}
which using \eqref{MGF-FOR-tk} yields the required formula.
\end{proof}

\subsection{Swaptions}

We now turn our attention to swaptions, but restrict ourselves to
one-dimensional affine processes as driving motions. The method for pricing
swaptions resembles \citeN{Jamshidian89} and has been also applied to
\lev-driven HJM models; cf. \citeN{EberleinKluge04}.

Recall that a payer (resp. receiver) swaption can be viewed as a put (resp.
call) option on a coupon bond with exercise price 1; cf. section 16.2.3 and
16.3.2 in \citeN{MusielaRutkowski97}. Consider a payer swaption with strike rate
$K$, where the underlying swap starts at time $T_k$ and matures at $T_m$
($k<m\le N$). The time-$T_k$ value is
\begin{align}
\mathbb{S}_{T_k}(K,T_k,T_m)
 = \left( 1-\sum^m_{i=k+1} c_i B(T_k,T_i)\right)^+,
\end{align}
where
\begin{align}%\label{}
c_i = \left\{
        \begin{array}{ll}
          \delta K, & \hbox{$k+1\le i\le m-1$,} \\
          1+\delta K, & \hbox{$i=m$.}
        \end{array}
      \right.
\end{align}
Now, we can express bond prices in terms of the martingales $M^u$, as follows:
\begin{align}%\label{}
B(T_k,T_i)
 &= \prod_{l=i}^{k-1} \frac{B(T_k,\tk[l+1])}{B(T_k,T_l)}
  = \prod_{l=i}^{k-1} \frac{M^{u_{l+1}}_{T_k}}{M^{u_{l}}_{T_k}}
  = \frac{M^{u_i}_{T_k}}{M^{u_k}_{T_k}},
\end{align}
since the product is again telescoping. Analogously to forward prices, cf.
\eqref{mg-3}, the dynamics of such quotients is again exponentially
affine:
\begin{align}%\label{}
\frac{M_{T_k}^{u_i}}{M_{T_k}^{u_k}}
 &= \exp\Big(\phi_{T_N-{T_k}}(u_i)- \phi_{T_N-{T_k}}(u_k)\nonumber\\
 &\qquad\qquad + 
  \big\langle\psi_{T_N-{T_k}}(u_i)-\psi_{T_N-{T_k}}(u_k),X_{T_k}\big\rangle\Big)
  \nonumber\\
 &=: \exp\big(A_{i,k}+B_{i,k}\cdot X_{T_k}\big).
\end{align}

Then, the time-0 value of the swaption is obtained by taking the discounted
$\pt[k]$-expectation of its time-$T_k$ value, hence
\begin{align}\label{swap-1}
\mathbb{S}_{0}(K,T_k,T_m)
 &= B(0,T_k)\,
    \E_{\pt[k]}\left[\left( 1-\sum^m_{i=k+1} c_i B(T_k,T_i)\right)^+\right]
\nonumber\\
 &= B(0,T_k)\,
    \E_{\pt[k]}\left[\left( 1-\sum^m_{i=k+1} c_i
\frac{M^{u_i}_{T_k}}{M^{u_k}_{T_k}}\right)^+\right] \nonumber\\
 &= B(0,T_k)\,
    \E_{\pt[k]}\left[\left( 1-\sum^m_{i=k+1} c_i \e^{A_{i,k}+B_{i,k}\cdot
X_{T_k}}\right)^+\right],
\end{align}
and this expectation can be computed with Fourier transform methods.

Define the functions $\underline{f}(x)=1-\sum^m_{i=k+1} c_i \exp\left(A_{i,k} +
B_{i,k}\cdot x\right)$ and
\begin{align}
f(x)
 = \max\{\underline{f}(x),0\}
 = \left( 1-\sum^m_{i=k+1} c_i \e^{A_{i,k}}\e^{B_{i,k}\cdot x}\right)^+.
\end{align}
We will also assume that, at least some, initial LIBOR rates are
\textit{positive}.

\begin{proposition}
The price of a swaption with strike rate $K$, option maturity $T_k$ and swap
maturity $T_m$ is given by
\begin{align}\label{swap-value}
\mathbb{S}_{0}(K,T_k,T_m)
 &= \frac{B(0,T_k)}{2\pi} \int_\R \Lambda_{X_{T_k}}(R-\ii v)\widehat{f}(v+ \ii
R)\ud v,
\end{align}
where the Fourier transform of the payoff function $f$ is
\begin{align}
\widehat{f}(v+\ii R)
 &= \e^{(iv-R)\mathcal{Y}}
    \left( \sum^m_{i=k+1}
\frac{c_i\e^{A_{i,k}+B_{i,k}\mathcal{Y}}}{B_{i,k}-R+ \ii v} - \frac{1}{\ii
v-R}\right).
\end{align}
Here $\mathcal{Y}$ denotes the unique zero of the function $\underline{f}$, the
$\pt[k]$-moment generating function $\Lambda_{X_{T_k}}$ of $X_{T_k}$ is given by
\eqref{Pk-mgf} and $R\in\mathcal{I}^k\cap(0,\infty)$.
\end{proposition}
\begin{proof}
Starting from \eqref{swap-1} and using Theorem 2.2 in
\shortciteN{EberleinGlauPapapantoleon08} again, we have that
\begin{align}\label{swap-value-2}
\mathbb{S}_{0}(K,T_k,T_m)
 &= B(0,T_k)\,
    \E_{\pt[k]}\left[\left( 1-\sum^m_{i=k+1} c_i \e^{A_{i,k}}\e^{B_{i,k}\cdot
X_{T_k}}\right)^+\right] \nonumber\\
 &= B(0,T_k)\, \int_\R
    \left( 1-\sum^m_{i=k+1} c_i \e^{A_{i,k}}\e^{B_{i,k}\cdot x}\right)^+
\P_{T_k,X_{T_k}}(\dx) \nonumber\\
 &= \frac{B(0,T_k)}{2\pi} \int_\R \Lambda_{X_{T_k}}(R-\ii v)\widehat{f}(v+\ii
R)\ud v,
\end{align}
where $\Lambda_{X_{T_k}}$ denotes the $\pt[k]$-moment generating function of the
random variable $X_{T_k}$, and $\widehat{f}$ denotes the Fourier transform of
the function $f$.

Now, we just have to calculate the Fourier transform of $f$ and show that the
prerequisites of the aforementioned theorem are satisfied. We know that
$\Lambda_{X_{T_k}}$ is finite for $R\in\mathcal{I}^k$. Since we have assumed
that some LIBOR rates are positive, Proposition \ref{initial-fit} and Lemma
\ref{positivity} yield that $B_{i,k} \le 0$ with strict inequality at least for
one $i \in \set{k+1,\dots,m}$. Hence, we can easily deduce that
$\underline{f}'(x)>0$, therefore $\underline{f}$ is a strictly increasing
function. Moreover, it is continuous and takes positive and negative values,
hence it has a unique zero, which we denote by $\mathcal{Y}$. Therefore,
\begin{align}\label{unique-zero-Y}
 f(x)=\underline{f}(x)1_{(\mathcal{Y},\infty)}.
\end{align}
Now, for $z\in\C$ with $\Im z>0$, the Fourier transform of $f$ is
\begin{align}
\widehat{f}(z)
 &= \int_\R \e^{\ii zx} \left( 1-\sum^m_{i=k+1} c_i \e^{A_{i,k}}\e^{B_{i,k}\cdot
x}\right)^+ \dx \nonumber\\
 &= \int^{\infty}_{\mathcal{Y}}\e^{\ii zx}
      \left( 1-\sum^m_{i=k+1} c_i \e^{A_{i,k}}\e^{B_{i,k}\cdot x}\right) \dx
\nonumber\\
 &= \int^{\infty}_{\mathcal{Y}}\e^{\ii zx}\dx
    - \sum^m_{i=k+1} c_i
\e^{A_{i,k}}\int^{\infty}_{\mathcal{Y}}\e^{(\ii z+B_{i,k})x}\dx\nonumber\\
 &= \e^{\ii z\mathcal{Y}}
    \left( \sum^m_{i=k+1} \frac{c_i\e^{A_{i,k}+B_{i,k}\mathcal{Y}}}{B_{i,k}+ \ii
z}-
\frac{1}{\ii z} \right).
\end{align}
Moreover, by examining the weak derivative of the dampened payoff function
$g(x)=\e^{-Rx}f(x)$ for $R>0$, we see that it is square integrable, as is $g$
itself. Hence $g$ lies in the Sobolev space $H^1(\R)$ and applying Lemma 2.5 in
Eberlein et al. \citeyear{EberleinGlauPapapantoleon08} yields that the Fourier
transform of $g$ is integrable.
\end{proof}

\begin{remark}
One method for pricing swaptions in multi-factor affine LIBOR models is to
follow the same procedure as above, where now $\mathcal{Y}$ will denote the set
of zeros of $f$ on $\R^d$. This computation can be challenging in general; but
see section \ref{2-CIR} for a case where it simplifies considerably. In the
following section we present another method for pricing swaptions, which is
interesting in its own right.
\end{remark}

\begin{remark}
In constrast to \citeN{Schrager_Pelsser_2006} we do not approximate swap rates
to overcome some analytical difficulties. This efficient technique could
certainly be applied, but this is not in the spirit of the present work. We
rather use the consequences of Axiom 3, namely the analytic tractability of any
vector of \lib rates with respect to any forward measure, to obtain pricing
formulas via Fourier pricing.
\end{remark}

\subsection{Swaptions in multi-factor models}

Next, we present an alternative method for pricing swaptions in affine LIBOR
models, which is particularly suitable for multi-factor models. The main idea is
to think of a swaption as a basket put option on artificial assets. Let $X$ be
an $\Rp^d$-valued affine process ($d>1$) and consider a swaption with
option maturity $T_k$ and swap maturity $T_k+T_i=T_m$. Accodring to
\eqref{swap-1} the price of this swaption of provided by
\begin{align}\label{swap-mf-ALM-1}
\mathbb{S}_{0}(K,T_k,T_m)
 &= B(0,T_k)\,\E_{\pt[k]} \left[\left( 1-\sum^m_{l=k+1} c_l \e^{A_{l,k} 
     + \scal{B_{l,k}}{X_{T_k}}} \right)^+\right] \nonumber\\
 &= B(0,T_k)\,\E_{\pt[k]} \left[\left( 1-\sum^m_{l=k+1} \e^{Y_l}
    \right)^+\right],
\end{align}
with the obvious definitions
\begin{align}
Y_l:= \gamma_l + \scal{B_{l,k}}{X_{T_k}}
\quad\text{ and }\quad
\gamma_l:= \log c_l + A_{l,k}.
\end{align}
The payoff function of the swaption resembles the payoff function of a basket
put option on $i$ assets, and the Fourier transform of the function
\begin{align}
g(x_1,\dots,x_i) = \left( 1 -\sum^i_{l=1} \e^{x_l} \right)^+
\end{align}
has been derived in \citeN{HubalekKallsen03}; see also \citeANP{HurdZhou09}
\citeyear{HurdZhou09}. We have that, for $z\in\C^i$ with $\Im z\le0$,
\begin{align}\label{swap-mf-ALM-2}
\widehat{g}(z) 
 = \frac{\prod_{l=1}^{i} \Gamma(\ii z_l)}
        {\Gamma\big( \ii\sum_{l=1}^{i} z_l+2\big)},
\end{align}
where $\Gamma$ denotes the Gamma function.

Therefore, we can derive a semi-analytical valuation formula for swaptions
applying Fourier methods again (cf.
\shortciteN[Theorem~3.2]{EberleinGlauPapapantoleon08}). Using
\eqref{swap-mf-ALM-1} and \eqref{swap-mf-ALM-2} we deduce that the price of a
swaption is provided by
\begin{align}\label{swaption-price}
\mathbb{S}_{0}(K,T_k,T_m)
 &= \frac{B(0,T_k)}{(2\pi)^i} \int_{\R^i}
      \widehat{g}(\ii R-v)M_Y(R+\ii v)\dv,
\end{align}
where the moment generating function $M_Y$ of $Y=(Y_{k+1},\dots,Y_m)$ can be
computed explicitly using the affine property of the driving factor process.
Indeed, for suitable $v\in\C^i$, we have that
\begin{align}
\E_{\pt[k]}\big[\exp\scal{v}{Y_t}\big]
 &= \e^{\scal{v}{\gamma_l}}\, \E_{\pt[k]}\left[\exp\left(\sum_{l=k+1}^m v_l
    \left(\sum_{j=1}^d B_{l,k}^j X_{T_k}^j\right)\right)\right] \nonumber\\
 &= \e^{\scal{v}{\gamma_l}}\, \E_{\pt[k]}\left[\exp\left(\sum_{j=1}^d 
    \left(\sum_{l=k+1}^m v_lB_{l,k}^j \right)X_{T_k}^j\right)\right] \nonumber\\
 &= \e^{\scal{v}{\gamma_l}}\, 
    \E_{\pt[k]}\big[\exp\left\langle U_k,X_{T_k}\right\rangle\big]\nonumber\\
 &= \exp\Big(\scal{v}{\gamma_l} + \phi^k_{T_k}(U_k) 
             + \scal{\psi^k_{T_k}(U_k)}{x}\Big),
\end{align}
where $U_k^j:=\sum_{l=k+1}^m u_lB_{l,k}^j$, while $\phi^k$ and $\psi^k$ are
provided by \eqref{Pk-mgf-0}.

\begin{remark}
One can immediately notice in \eqref{swaption-price} that the dimension of the
integration depends on the length of the underlying swap and not on the
dimension of the factor process, as is the case in the one-factor model, cf.
\eqref{swap-value}. The \textit{disadvantage} here is that the formula becomes
infeasible even when the swap is moderately long (e.g. an $x$-on-3 years
swaption in a semiannual tenor structure requires a 6-dimensional integration).
The \textit{advantage} is that, for contracts with short swap maturity, the
swaption formula yields explicit results irrespective of how many factors the
model has. In any case, one can use the exact formula \eqref{swaption-price} as
a \textit{benchmark} to test faster approximative solutions (see, for example,
\citeNP{CollinDufresne_Goldstein_2002}).
\end{remark}

\subsection{Hedging and Greeks}

Hedging interest rate derivatives in this class of models will be dealt with in
future research, however we would like to point out some important aspects. On
the one hand, option price sensitivities, i.e. Greeks, can be calculated in the
affine LIBOR model in semi-analytical form. Indeed, under certain conditions we
are allowed to exchange integration and differentiation of the option price
formulas, which leads to Fourier-based methods for Greeks (cf. e.g.
\shortciteNP[section 4]{EberleinGlauPapapantoleon08}).

On the other hand, certain examples of affine \lib models are \textit{complete}
in their own filtration, whence the hedging strategy is provided by
$\Delta$-hedging. The CIR model -- presented in section \ref{CIR} -- is such an
example, where completeness follows from the continuity of the paths and the
Markov property.

\section{Examples}
\label{toy-1}

We present here three concrete specifications of the affine LIBOR models we have
constructed. In the first two specifications the driving processes are the
Cox-Ingersoll-Ross process and an OU-type process driven by a compound Poisson
subordinator with exponentially distributed jumps, such that it has the Gamma
law as stationary distribution. The third specification is a 2-factor extension
of the CIR-driven model. We first describe the driving affine processes and then
discuss the affine martingales used to model \lib rates. In the case of the
1-factor CIR driving process we derive closed-form pricing formulas for caps and
swaptions, utiling the $\chi^2$-distribution function. Moreover, in the 2-factor
extension of the CIR model a closed form pricing formula for caps is obtained,
while a semi-analytical formula for swaptions is also derived.

\subsection{CIR martingales}
\label{CIR}

The first example is the Cox-Ingersoll-Ross (CIR) process, given by
\begin{align}\label{Bessel-sde}
 \dd X_t = -\lambda \left(X_t - \theta \right)\dt + 2\eta \sqrt{X_t}\ud W_t,
 \quad X_0=x\in\Rp,
\end{align}
where $\lambda, \theta, \eta \in \Rp$. This process is an affine process on
$\Rp$, with
\begin{align}\label{Bessel-FR}
 F(u) = \lambda \theta  u
\quad\text{ and }\quad
 R(u) = 2 \eta^2 u^2 - \lambda u.
\end{align}
Its moment generating function is given by
\begin{align}\label{Bessel-mgf}
\E_x\big[\e^{uX_t}\big]
 &= \exp\Big(\phi_t(u) + x \cdot \psi_t(u)\Big),
\end{align}
where
\begin{align}\label{Bessel-phipsi}
 \phi_t(u) = -\frac{\lambda \theta}{2 \eta^2}\log\big(1-2\eta^2 b(t) u\big)
\quad\text{ and }\quad
 \psi_t(u) = \frac{a(t)u}{1 - 2 \eta^2 b(t) u},
\end{align}
with
\begin{align*}
b(t) =
\begin{cases}
 t, \; & \text{if}\;\lambda = 0\\
 \frac{1 - \e^{-\lambda t}}{\lambda},\; & \text{if}\;\lambda\neq0
\end{cases},
\qquad \text{and} \qquad
a(t) = \e^{-\lambda t}.
\end{align*}
The martingales defined in \eqref{mg-1} thus take the form
\begin{align}\label{Bessel-mg}
M_t^u
 &= \exp\big(\phi_{T_N-t}(u) + \langle\psi_{T_N-t}(u),X_t\rangle\big)
\\\nonumber
 &= \exp\left( -\frac{\lambda\theta}{2\eta^2}\log\big(1-2\eta^2 b(T_N-t)u\big)
               + \frac{\e^{-\lambda(T_N-t)}u}{1-2\eta^2 b(T_N-t)u}\cdot X_t
\right),
\end{align}
where $u$ must be chosen such that $u < \frac{1}{2\eta^2 b(T_N)}$. Note that
$\gamma_X = \infty$ (see Definition~\ref{def-gamma}), such that by
Proposition~\ref{initial-fit} the model can fit any term structure of initial
LIBOR rates.\\

In order to describe the marginal distribution of this process, we derive some
useful results on an extension of the non-central chi-square distribution. We
say that a random variable $Y$ has \textit{location-scale extended non-central
chi-square} distribution with parameters $(\mu,\sigma,\nu,\alpha)$, or short
$Y\sim\lsnc(\mu,\sigma,\nu,\alpha)$, if $\frac{Y - \mu}{\sigma}$ has non-central
chi-square distribution with $\nu$ degrees of freedom and non-centrality
parameter $\alpha$. The density and distribution function of $Y$ can be derived
in the obvious way from the density and distribution function of the non-central
chi-square law. We will also need the cumulant generating function of $Y$, which
is given by
\begin{align}
\kappa_{\lsnc} (u) &= -\frac{\nu}{2} \log\left(1 - 2\sigma
u\right) + \frac{\alpha \sigma u}{1 - 2\sigma u} + \mu u, \qquad
(u < \frac{1}{2\sigma})\;.\label{cgf-lsnc}
\end{align}
For any $\vartheta < \frac{1}{2\sigma}$ we may consider the random variable
$Y_\vartheta$ with distribution function $F_\vartheta$, defined through the
exponential change of measure
$\frac{\ud F_\vartheta}{\ud F} = \e^{x \vartheta - \kappa(\vartheta)}$. It is
well known that the cumulant generating function of $Y_\vartheta$ is given by
$\kappa_\vartheta(u) = \kappa(u + \vartheta) - \kappa(\vartheta)$. For the
$\lsnc$-distribution a simple calculation using \eqref{cgf-lsnc} shows that
\begin{align}\label{lsnc-measurechange}
Y_{\vartheta} \sim
\lsnc\left(\mu,\frac{\sigma}{\zeta},\nu,\frac{\alpha}{\zeta}\right),
\qquad \text{with} \qquad \zeta = 1 - 2\sigma\vartheta > 0.
\end{align}

Let us now return to the CIR process $X$. Comparing \eqref{Bessel-phipsi} and
\eqref{cgf-lsnc}, shows that
\begin{align}
X_t \; \stackrel{\P_{T_N}}{\sim} \; \lsnc\left(0,\eta^2 b(t), \frac{\lambda
\theta}{\eta^2},
\frac{x a(t)}{\eta^2 b(t)}\right)\,,
\end{align}
i.e. the marginals of $X$ have $\lsnc$-distribution under the terminal measure.
By \eqref{Pk-mgf-2} and \eqref{Pk-mgf-0} we know that the measure change from
the
terminal measure $\P_{T_N}$ to the forward measure $\P_{T_{k+1}}$ is an
exponential change of measure, with $\vartheta = \psi_{T_N-t}(u_{k+1})$. Thus,
we derive from \eqref{lsnc-measurechange} that
\begin{align}\label{LSNC-Pk}
X_t \; &\stackrel{\P_{T_k}}{\sim} \;
\lsnc\left(0,\frac{\eta^2 b(t)}{\zeta(t,\tk[k])},
           \frac{\lambda \theta}{\eta^2},
           \frac{x a(t)}{\eta^2 b(t) \zeta(t,\tk[k])}\right),
\end{align}
where
\begin{align}
\zeta(t,\tk[k]) = 1 - 2 \eta^2 b(t) \psi_{T_N-t}(u_k).
\end{align}
Finally, it follows from \eqref{simple}, that the log-forward rates have
distribution
\begin{align}
\log\left(\frac{M_{T_k}^{u_k}}{M_{T_k}^{\uk}}\right) \;
&\stackrel{\P_{T_{k+1}}}{\sim} \;
\lsnc\left(A_k,\frac{B_k \eta^2 b(T_k)}{\zeta(T_k,\tk[k+1])},
 \frac{\lambda\theta}{\eta^2},
 \frac{x a(T_k)}{\eta^2 b(T_k) \zeta(T_k,\tk[k+1])}\right)
\end{align}
under the corresponding forward measure, where $A_k$ and $B_k$ are given by
\eqref{simple}. Hence, log-forward rates are LSNC$-\chi^2$-distributed under
\textit{any} forward measure with different parameters $\sigma$ and $\alpha$,
due to the different $\zeta$.

We are now in the position to derive a closed-form caplet valuation formula for
the CIR model. Denoting by
$Z =\log\left(\frac{B(T_k,T_k)}{B(T_k,T_{k+1})}\right)$ the log-forward rate, it
holds that
\begin{align}
\mathbb{C}(T_k,K)
 &= B(0,\tk[k+1])\,
    \E_{\P_{\tk[k+1]}}\left[\left(\e^Z - \mathscr{K}\right)^+\right]\\
 &= B(0,\tk[k+1])\,
    \left\{\E_{\P_{\tk[k+1]}}\left[\e^Z 1_{\set{Z\ge\log\mathscr{K}}}\right]
           - \mathscr{K}\,\P_{\tk[k+1]}\left[Z\ge\log\mathscr{K}\right] \right\}
\nonumber \\
 &= B(0,\tk[k])\,\P_{\tk[k]}\left[Z\ge\log\mathscr{K}\right]
           -
\mathscr{K}\,B(0,\tk[k+1])\,\P_{\tk[k+1]}\left[Z\ge\log\mathscr{K}\right]\;,
\nonumber
\end{align}
where we have used \eqref{Pk-to-next} and $\mathscr{K}=1+\delta K$. The
probability terms can be evaluated through the distribution function of the
$\lsnc$-distribution. After some calculations, we arrive at the following
result:
\begin{equation}\label{CIR-caplet-closed}
\mathbb{C}(T_k,K)
 = B(0,T_k)\cdot
\overline{\chi}^2_{\nu,\alpha_1}\left(\frac{\log\mathscr{K}-A_k}{B_k \sigma_1}
\right)
 - \mathscr{K}^\star \cdot
\overline{\chi}^2_{\nu,\alpha_2}\left(\frac{\log\mathscr{K}-A_k}{B_k \sigma_2}
\right)\,,
\end{equation}
where  $\mathscr{K}^\star:=\mathscr{K}\,B(0,\tk[k+1])$, 
$\overline{\chi}^2_{\nu,\alpha}(x) := 1 -\chi^2_{\nu,\alpha}(x)$,
with $\chi^2_{\nu,\alpha}(x)$ the non-central chi-square distribution
function, and
\begin{align*}
\nu = \frac{\lambda\theta}{\eta^2}, \quad
\sigma_i = \frac{\eta^2 b(T_k)}{\zeta(T_k,T_i)} \quad \text{and} \quad
\alpha_i = \frac{x a(T_k)}{\eta^2 b(T_k) \zeta(T_k,T_i)},
\end{align*}
for each $i\in\{k, \dotsc, N\}$.

Similarly a closed-form pricing formula for swaptions can be derived; by
\eqref{swap-1}, using \eqref{Pk-to-final}, we get
\begin{align}
&\mathbb{S}_{0}(K,T_k,T_m) \nonumber\\
 &= B(0,T_k)\,
    \E_{\pt[k]}\left[\left( 1-\sum^m_{i=k+1} c_i
\frac{M^{u_i}_{T_k}}{M^{u_k}_{T_k}}\right)^+\right] \nonumber\\
 &= B(0,T_N)\,\E_{\pt[N]}\left[\left(M_{T_k}^{u_k} - \sum^m_{i=k+1} c_i
 M^{u_i}_{T_k}\right)^+\right]\nonumber\\
 &= B(0,T_N)
    \Big\{ \E_{\pt[N]}\left[M_{T_k}^{u_k}1_{\{X_{T_k}\ge\mathcal{Y}\}}\right]
           - \sum^m_{i=k+1} c_i
\E_{\pt[N]}\left[M^{u_i}_{T_k}1_{\{X_{T_k}\ge\mathcal{Y}\}}\right]\Big\},
\end{align}
where $\mathcal{Y}$ is defined as in \eqref{unique-zero-Y}. Using the known
distribution function of $X_{T_k}$ under $\P_{T_k}$, cf. \eqref{LSNC-Pk}, the
exponential change of measure formula \eqref{lsnc-measurechange} and
\eqref{Pk-to-final} once again, we arrive at
\begin{align}
\mathbb{S}_{0}(K,T_k,T_m)
 &= B(0,T_k) \cdot
    \overline{\chi}^2_{\nu,\alpha_k}\left(\mathcal{Y}/\sigma_k\right)
   - \sum_{i=k+1}^m c_i B(0,T_i)\cdot
    \overline{\chi}^2_{\nu,\alpha_i}\left(\mathcal{Y}/\sigma_i\right).
\end{align}

\subsection{$\Gamma$-OU martingales}

The second example is an OU-process on $\Rp$ such that the limit law is the
Gamma distribution. Consider the SDE
\begin{align}
\ud X_t = -\lambda X_t\dt + \ud H_t,
 \quad X_0=x\in\Rp,
\end{align}
where $\lambda>0$. The driving \lev process \prazess[H]is a compound Poisson
subordinator with cumulant generating function
\begin{align}
\kappa_{\text{CP}}(u)
 = \frac{\lambda\beta u}{\alpha-u},\qquad(u < \alpha)
\end{align}
where $\alpha,\beta>0$. That is, $H$ is a compound Poisson process with jump
intensity $\lambda\beta$ and exponentially distributed jumps with parameter
$\alpha$. The moment generating function of $H$ is well defined for
$u\in\mathcal{I}_{\text{CP}} = (-\infty,\alpha)$. The limit law of this OU
process is the Gamma distribution $\Gamma(\alpha,\beta)$, i.e. it has the
cumulant generating function
\begin{align}
 \kappa_{\Gamma}(u) = - \beta\ln\big(1-\frac{u}{\alpha}\big);
\end{align}
cf. Theorem 3.15 in \citeN{KellerResselSteiner08} or Nicolato and Venardos
\citeyear{NicolatoVenardos03}. We call the resulting affine process the
$\Gamma$-OU process.

The moment generating function of the random variable $X_t$, using Lemma 17.1 in
\citeN{Sato99}, is
\begin{align}\label{mgf-X-cp}
\E_x\big[\e^{uX_t}\big]
 &= \exp\left( \int\nolimits_0^t\kappa_{\text{CP}}(\e^{-\lambda s}u)\ds
             + \e^{-\lambda t}u\cdot x\right).
\end{align}
Using the change of variables $y=\e^{\lambda s}$ and
$\int\frac{1}{x(ax+b)}\dx=-\frac{1}{b}\ln|\frac{ax+b}{x}|$, we get:
\begin{align}
\int\nolimits_0^t\kappa_{\text{CP}}(\e^{-\lambda s}u)\ds
 &= \int\nolimits_0^t \frac{\lambda\beta\e^{-\lambda s}u}{\alpha-\e^{-\lambda
s}u} \ds
  = \beta \ln \left(\frac{\alpha-\e^{-\lambda t}u}{\alpha-u}\right),
\end{align}
since $u\in(-\infty,\alpha)$. Hence, the moment generating function in
\eqref{mgf-X-cp} is
\begin{align}\label{mgf-X-cp-final}
\E_x\big[\e^{uX_t}\big]
 &= \exp\left( \beta \ln\left(\frac{\alpha-\e^{-\lambda t}u}{\alpha-u}\right)
             + \e^{-\lambda t}u\cdot x\right),
\end{align}
which yields that $X$ is an affine process on $D=\Rp$ with
\begin{align}\label{cp-phipsi}
 \phi_t(u) = \beta \ln\left(\frac{\alpha-\e^{-\lambda t}u}{\alpha-u}\right)
\quad\text{ and }\quad
 \psi_t(u) = \e^{-\lambda t}u,
\end{align}
and the functions $F$ and $R$ have the form
\begin{align}\label{cp-FR}
 F(u) = \frac{\lambda\beta u}{\alpha - u}
\quad\text{ and }\quad
 R(u) = -\lambda u.
\end{align}

Therefore, the affine martingales constructed in \eqref{mg-1} take now the form
\begin{align}\label{mg-CP-OU}
M_t^u
 &= \exp\big(\phi_{T_N-t}(u) + \langle\psi_{T_N-t}(u),X_t\rangle\big)
\nonumber\\
 &= \exp\left( \beta\ln\left(\frac{\alpha-\e^{-\lambda
(T_N-t)}u}{\alpha-u}\right)
              + \e^{-\lambda (T_N-t)}u\cdot X_t  \right),
\end{align}
where $u$ must be chosen such that
$u\in\mathcal{I}_{\text{CP}}\cap\Rp=[0,\alpha)$. Moreover, we have that
$\gamma_X=\infty$, hence the model can fit any term structure of initial \lib
rates.

\subsection{2-factor CIR martingales}
\label{2-CIR}

An affine LIBOR model with a single-factor driving process might not be
completely satisfactory from an econometric point of view. Consider, for
example, the instantaneous correlation between log-forward rates of different
maturities. By \eqref{simple} the log-forward rates depend linearly on the
driving process $X$, such that in the case of a single-factor model they are
perfectly correlated. `Decorrelation' can be achieved by adding additional
factors (see Remark~6.3.1 in Brigo and Mercurio \citeyearNP{BrigoMercurio06}), 
which in general
improves the econometric characteristics of the model in other aspects too, e.g.
by allowing a more flexible term structure of forward rate volatilities. 

Here we propose a simple $2$-factor extension of the CIR model presented in
Section~\ref{CIR} by adding a second independent CIR process. Note that although
the factors are independent, the specification of forward rates through
\eqref{simple} will lead to non-trivial correlations between forward and LIBOR
rates of different maturities. As we shall see, we still have a closed form
valuation formula for caplets and a very tractable formula for swaptions in the
extended model. The driving process $X$
consists of the two factors $X^1,X^2$ given for $j \in \set{1,2}$ by
\begin{equation}\label{Bessel-sde2}
 \dd X^j_t = -\lambda_j \left(X^j_t - \theta_j \right)\dt + 2\eta_j
\sqrt{X^j_t}\ud W^j_t,
 \quad X^j_0=x_j\in\Rp,
\end{equation}
where $\lambda_j, \theta_j, \eta_j \in \Rp$ and the Brownian motions $W^1$ and
$W^2$ are independent. This process is an affine process on
$\Rp^2$, with
\begin{align}\label{Bessel-FR2}
 F(u) = \lambda_1 \theta_1  u^1 + \lambda_2 \theta_2 u^2
\quad\text{ and }\quad
 R_j(u) = 2 \eta_j^2 (u^j)^2 - \lambda_j u^j,
\end{align}
where $u = (u^1,u^2)$. Define $\phi^j_t(u), \psi^j_t(u)$ as in
\eqref{Bessel-phipsi}, by adding the index $j \in \set{1,2}$ to all parameters,
including $u$. Furthermore set $\phi_t(u) = \phi_t^1(u^1) + \phi_t^2(u^2)$. The
martingales $M^u$ now take the form
\begin{align}\label{Bessel-mg2}
M_t^{u} 
 = \exp\Big(\phi_{T_N-t}(u) + \psi^1_{T_N-t}(u^1) \cdot X^1_t
    + \psi^2_{T_N-t}(u^2) \cdot X^2_t\Big),
\end{align}
where $u^j$ must be chosen such that $u^j < \frac{1}{2\eta_j^2 b_j(T_N)}$. Note
that the $u_k = (u_k^1, u_k^2)$ are no longer uniquely determined by fitting the
initial LIBOR rates, but provide additional freedom to fit e.g. the term
structure of caplet implied volatilities or the implied correlations from
swaption prices.

It is clear that in analogy to Section~\ref{CIR} $X_t^1$ and $X_t^2$ follow an
$\lsnc$-distribution under all forward measures. The log-forward rate can be
written as $Z = A_k + B_k^1 X_{T_k}^1 + B_k^2 X_{T_k}^2$, cf. \eqref{simple},
where 
\begin{align*}
A_k = \phi_{T_N - T_k}(u_k) - \phi_{T_N - T_k}(u_{k+1}), \qquad
B_k^j = \psi^j_{T_N - T_k}(u^j_k) - \psi^j_{T_N - T_k}(u^j_{k+1}) 
\end{align*}
for $j \in \set{1,2}$. Thus the distribution of $Z$ is a convolution of two
$\lsnc$-distributions which, however, does not belong to the same  class.
Nevertheless $Z$ can be expressed as a positive linear combination of
independent non-central chi-squared random variables plus a (deterministic)
constant. Distributions of this type have been studied extensively in the
context of quadratic forms of normal random variables. Their distribution
function has an infinite series expansion (see~\citeNP{Press66}) and has been
implemented in several software packages (e.g. \texttt{CompQuadForm} in the
statistical computation environment $\textsf{R}$; see
\citeNP{DuchesneMicheaux10}). For vectors $\boldsymbol{\sigma},
\boldsymbol{\nu}, \boldsymbol{\alpha}$ of $n$ positive elements each, let
$U(x;\boldsymbol{\sigma}, \boldsymbol{\nu}, \boldsymbol{\alpha})$ denote the
distribution function of $\sum_{j=1}^n \sigma_j Y_j$ where $Y_j$ are independent
non-central chi-square distributed random variables with $\nu_j$ degrees of
freedom and non-centrality parameter $\alpha_j$ respectively. In addition,
define $\overline{U}(x) = 1 - U(x)$. Then, in analogy to
\eqref{CIR-caplet-closed}, we obtain the following closed-form caplet valuation
formula: 
\begin{align}
\mathbb{C}(T_k,K)
 &= B(0,T_k)\cdot \overline{U}\big(\log\mathscr{K}-A_k;
\boldsymbol{\tilde \sigma}_{k}, \boldsymbol{\nu}, \boldsymbol{\alpha}_{k}\big)
 \nonumber \\
 &\quad- \mathscr{K}^\star \cdot \overline{U}\big(\log\mathscr{K}-A_k;
  \boldsymbol{\tilde
\sigma}_{k+1},\boldsymbol{\nu},\boldsymbol{\alpha}_{k+1}\big)
\end{align}
where $\mathscr{K}^\star=\mathscr{K}\,B(0,\tk[k+1])$, $\boldsymbol{\nu} =
\left(\tfrac{\lambda_1 \theta_1}{\eta_1^2}, \tfrac{\lambda_2
\theta_2}{\eta_2^2}\right)$,  $\boldsymbol{\tilde \sigma_k} = (B_k^1 \sigma_k^1,
B_k^2 \sigma_k^2)$ and $\boldsymbol{\alpha_k} = (\alpha_k^1, \alpha_k^2)$  with
\begin{align*}
\sigma_i^j &= \left(\eta_j^2 b_j(T_k)\right) / \zeta_i^j, \\
\alpha_i^j &= x_j a_j(T_k) / \left(\eta_j^2 b_j(T_k) \zeta_i^j\right),
\\
\zeta_i^j &= 1 - 2\eta_j^2 b_j(T_k)\psi^j_{T_N - T_k}(u^i_k)
\end{align*}
for $j \in \set{1,2}$ and $i \in \set{k, \dotsc, N}$.\\

Regarding the valuation of swaptions, define
\begin{align*}
A_{k,i} = \phi_{T_N - T_k}(u_i) - \phi_{T_N - T_k}(u_k), \qquad
B_{k,i}^j = \psi^j_{T_N - T_k}(u^j_i) - \psi^j_{T_N - T_k}(u^j_k),
\end{align*} 
consider the function
\[\underline{f}(x,y) = 1 - \sum_{i=k+1}^m c_i \exp \left(A_{k,i} + B_{k,i}^1  x
+ B_{k,i}^2 y \right),\]
and for each $x \in \Rp$ define 
\[\mathcal{Y}(x) = \inf \set{y \ge 0: \underline{f}(x,y) = 0}.\]
Since $\underline{f}(x,y)$ is continuous and increasing in $y$ for each fixed
$x$ we have that $f(x,y) = (\underline{f}(x,y))^+$ can be written as $f(x,y) =
\underline{f}(x,y)\Ind{y \ge \mathcal{Y}(x)}$.
Hence, the price of a swaption $\mathbb{S}_{0}(K,T_k,T_m)$ can be written as
\begin{align*}
 \mathbb{S}_{0}(K,T_k,T_m) &= B(0,T_k)
\E_{\pt[k]}\left[f(X_{T_k}^1,X_{T_k}^2)\right] \\
 &= B(0,T_k) \E_{\pt[k]}\left[\underline{f}(X_{T_k}^1,X_{T_k}^2) \Ind{X^2_{T_k}
\ge \mathcal{Y}(X^1_{T_k})}\right] \\
 &= B(0,T_k) \pt[k]\left[X^2_{T_k} \ge \mathcal{Y}(X^1_{T_k})\right] \\
 &\quad - \sum_{i = k+1}^m c_i B(0,T_i) \pt[i]\left[X^2_{T_k} \ge
\mathcal{Y}(X^1_{T_k})\right].
\end{align*}
Using the fact that $X^2_{T_k}$ is $\lsnc$-distributed under each forward
measure, we can further rewrite the above as
\begin{align*}
 \mathbb{S}_{0}(K,T_k,T_m) &= B(0,T_k)
\E_{\pt[k]}\left[\overline{\chi}^2_{\nu^2,\alpha^2_k}\left(\frac{\mathcal{Y}
(X^1_{T_k})}{\sigma^2_k}\right)\right] \\
 &\quad - \sum_{i = k+1}^m c_i B(0,T_i)
\E_{\pt[i]}\left[\overline{\chi}^2_{\nu^2,\alpha^2_i}\left(\frac{\mathcal{Y}
(X^1_{T_k})}{\sigma^2_i}\right)\right] =: \mathbb{S}_{0}.
\end{align*}
Switching back to the terminal measure, we finally obtain
\begin{align}
\mathbb{S}_{0} &= B(0,T_N) \E_{\pt[N]}\left[\beta_{k,k} \exp
 \left(\psi^1_{T_N - T_k}(u_k^1)
 X^1_{T_k}\right)\overline{\chi}^2_{\nu^2,\alpha^2_k}
\left(\frac{\mathcal{Y}(X^1_{T_k})}{\sigma^2_k}\right)\right. \nonumber \\
 &\quad - \left. \sum_{i = k+1}^m c_i \beta_{k,i} \exp 
 \left(\psi^1_{T_N - T_k}(u_i^1) X^1_{T_k}\right)
 \overline{\chi}^2_{\nu^2,\alpha^2_i}\left(\frac{\mathcal{Y}(X^1_{T_k})}
  {\sigma^2_i}\right)\right],
\end{align}
where 
\begin{align*}
\beta_{k,i} = 
 \exp\left(\vphantom{\int}\phi^1_{T_N - T_k}(u_i^1) +
  \phi_{T_N-T_K}^2(u_i^2) + \psi_{T_N-T_k}^2(u_i^2)x_2\right).
\end{align*}
This expression can be evaluated by a one-dimensional numerical integration
against the non-central chi-square distribution, while the value of
$\mathcal{Y}(x)$ has to be computed by a line search.

\section{Numerical illustration}
\label{numerical_illustrations}

In order to showcase some prototypical volatility surfaces resulting from the
proposed affine \lib models, we consider a tenor structure of zero coupon bond
prices generated from the Svensson family. We fit the initial LIBOR rates
implied by the bond prices using the $u$'s as described in Proposition
\ref{initial-fit}, and then price caplets and plot the implied volatility
surfaces for different parameters of the driving affine factor process. The
implied volatility surface corresponding to the CIR parameters
\begin{align*}
\lambda = 0.026, \quad \theta = 0.65, \quad \eta  = 0.5 \quad
\text{and} \quad X_0 = 3.45
\end{align*}
is shown in Figure \ref{fig:CIR_volsurface}. The implied volatilities from the
CIR model exhibit a skewed shape as a function of strike price, while the
term structure is decreasing as a function of maturity. In the example
corresponding to the $\Gamma$-OU process, we consider the same tenor structure
and the implied volatility surface corresponding to the $\Gamma$-OU parameters
\begin{align*}
\lambda = .05, \quad \alpha = 0.8, \quad \beta = 0.5, \quad X_0 = 1.35
\end{align*}
is shown in Figure \ref{fig:GammaOU_volsurface}. The implied volatility exhibits
a smile shape in this example, while we can still observe a decreasing  behavior
for the term structure of volatilities. 
\begin{figure}
 \centering
 \includegraphics[width=10cm]{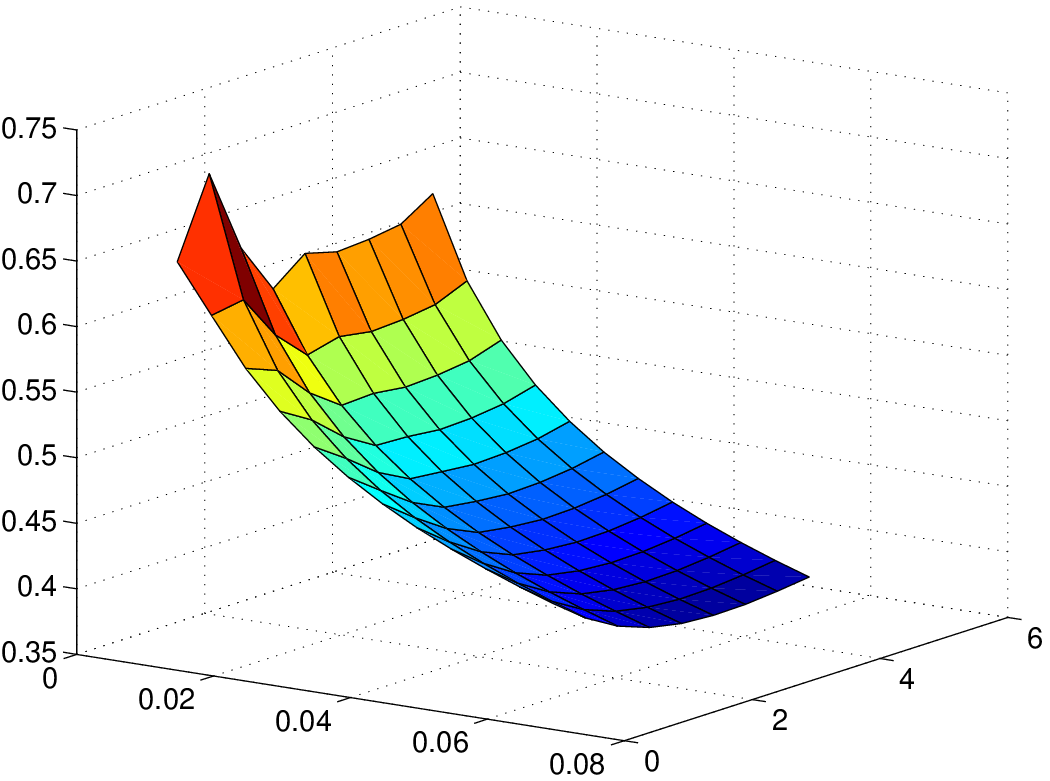}
 \caption{Implied volatility surface for the CIR martingales.}
 \label{fig:CIR_volsurface}
\end{figure}
\begin{figure}
 \centering
 \includegraphics[width=10cm]{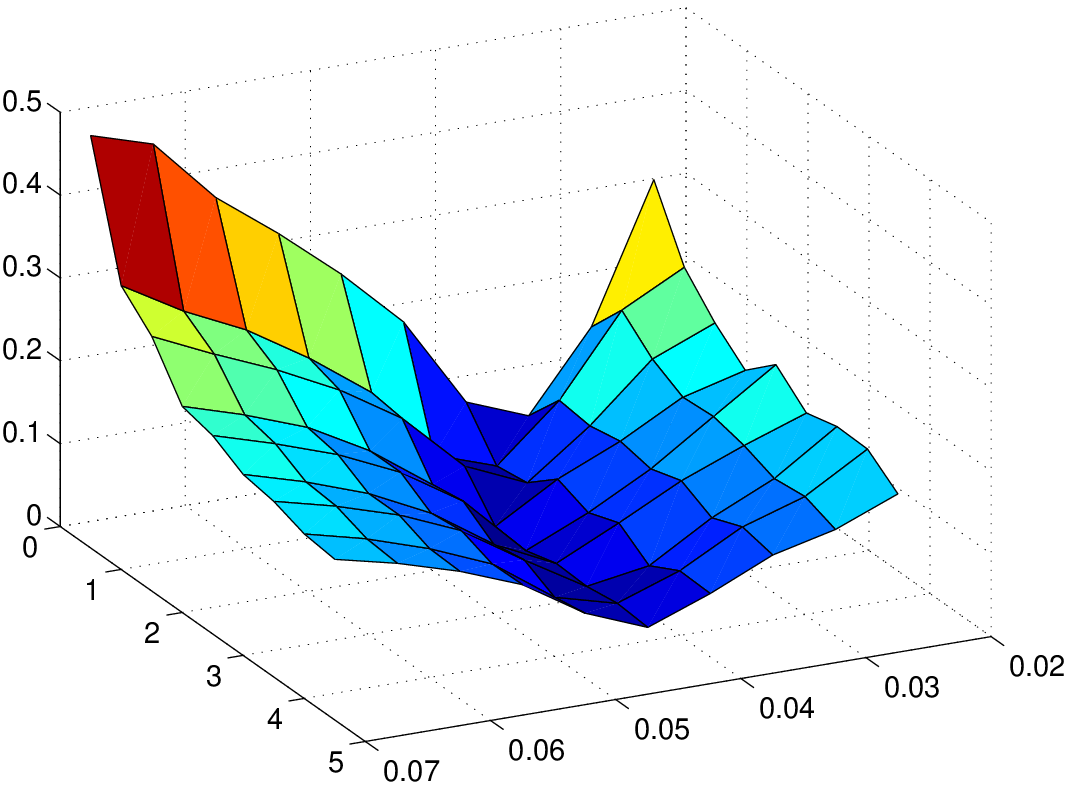}
 \caption{Implied volatility surface for the $\Gamma$-OU martingales.}
 \label{fig:GammaOU_volsurface}
\end{figure}
Finally, let us point out that these figures  stem from 1-factor models; one
would expect to observe even flexible shapes of surfaces from multi-factor
affine LIBOR models. Several impressive results in this direction have been
obtained by \citeANP{DaFonseca_Gnoatto_Grasselli_2011}
\citeyear{DaFonseca_Gnoatto_Grasselli_2011} .

\bibliographystyle{chicago}
\bibliography{references}

\end{document}